\pdfoutput=1

\documentclass{article}
\usepackage{kr}

\usepackage{times}
\usepackage{soul}
\usepackage{url}
\usepackage[hidelinks]{hyperref}
\usepackage[utf8]{inputenc}
\usepackage[small]{caption}
\usepackage{graphicx}
\usepackage{amsmath}
\usepackage{amsthm}
\usepackage{booktabs}
\usepackage{algorithm}
\usepackage{algorithmic}
\usepackage{xspace}
\usepackage{tikz}
\usetikzlibrary{shapes,arrows,backgrounds,fit,positioning}
\usepackage{multicol}
\usepackage{caption,subcaption}
\usepackage{amssymb}
\usepackage{soul}
\usepackage{thmtools,thm-restate}

\theoremstyle{definition}
\newtheorem{example}{Example}
\newtheorem*{example*}{Example}

\newtheorem*{theorem*}{Theorem}
\newtheorem{definition}{Definition}
\newtheorem*{definition*}{Definition}
\newtheorem{proposition}{Proposition}
\newtheorem*{proposition*}{proposition}
\newtheorem{corollary}{Corollary}
\newtheorem*{corollary*}{Corollary}

\newtheorem*{lemma*}{Lemma}

\newcommand{\showproofs}{1}

\usepackage{enumitem}
\setlist[itemize]{noitemsep, topsep=0pt}
\setlist[enumerate]{noitemsep, topsep=0pt}
\usepackage[page]{appendix} 

\title{A Framework for Combining Entity Resolution and\\ Query Answering in Knowledge Bases}

\author{%
Ronald Fagin$^1$\and
Phokion G.\ Kolaitis$^{1,2}$\and
Domenico Lembo$^3$\and
Lucian Popa$^1$ \and
Federico Scafoglieri$^3$\\
\affiliations
$^1$IBM Research, Almaden, USA\\
$^2$UC Santa Cruz, USA\\
$^3$Sapienza University of Rome, Italy\\
\emails
fagin@us.ibm.com,
kolaitis@ucsc.edu,
lembo@diag.uniroma1.it,\\
lpopa@us.ibm.com,
scafoglieri@diag.uniroma1.it
}

\newcommand{\commentout}[1]{}

\newcommand{\chase}{\mathit{chase}}

\newcommand{\id}{\mathsf{id}}

\newcommand{\valueSet}{V}
\newcommand{\entitySet}{E}
\newcommand{\genericSet}{T}
\newcommand{\genericSettwo}{U}
\newcommand{\trasf}{\tau}
\newcommand{\setvar}{\mathit{SetVar}}
\newcommand{\dareduced}{\downarrow^{\rho}}

\newcommand{\DB}{\textbf{D}}

\newcommand{\I}{\mathcal{I}} 
\newcommand{\K}{\mathcal{K}}

\renewcommand{\S}{\mathcal{S}} \newcommand{\T}{\mathcal{T}}
\newcommand{\U}{\mathcal{U}} \newcommand{\V}{\mathcal{V}}
 
 \newcommand{\Z}{\mathcal{Z}}

\newcommand{\ra}{\rightarrow}

\newcommand{\tup}[1]{\langle #1\rangle}            

\newcommand{\activeD}{\mathit{active}}
\newcommand{\underD}{\mathit{under}}

\def\-{\raisebox{.75pt}{-}}

\renewcommand{\tup}[1]{\langle #1\rangle}

\newcommand{\cert}{\mathit{cert}}

\renewcommand{\DB}{\mathbf{D}}

\newcommand{\Mod}{\mathit{Mod}}

\newcommand{\sig}{\mathit{sig}}

\newcommand{\alphaOntology}{\S_O}
\newcommand{\alphaPredicates}{\S_P}
\newcommand{\alphaBuiltIns}{\S_B}
\newcommand{\alphaEntities}{\S_E}
\newcommand{\alphaValues}{\S_V}
\newcommand{\alphaVariables}{\S_\V}

\newcommand{\alphaEntityNulls}{\S_{EN}}
\newcommand{\alphaValueNulls}{\S_{VN}}

\newcommand{\arity}{\mathit{arity}}

\newcommand{\type}{\mathit{type}}
\newcommand{\entity}{\mathsf{e}}
\newcommand{\val}{\mathsf{v}}

\newif\ifdraft
\drafttrue
\ifdraft
\newcommand{\nb}[1]{\textcolor{red}{\bf!}%
	\marginpar[\parbox{15mm}{\raggedleft\scriptsize\textcolor{red}{#1}}]%
	{\parbox{15mm}{\raggedright\scriptsize\textcolor{red}{#1}}}}
\else
\newcommand{\nb}[1]{}
\fi

\begin{document}

\maketitle

\begin{abstract}
We propose a new framework for combining entity resolution and query answering in knowledge bases (KBs) with tuple-generating dependencies (tgds) and equality-generating dependencies (egds) as rules.
We define the semantics of the KB in terms of special instances that involve equivalence classes of entities and sets of values. Intuitively, the former collect all entities denoting the same real-world object, while the latter collect all alternative values for an attribute. This approach allows us to both resolve entities and bypass possible inconsistencies in the data. We then design a chase procedure that is tailored to this new framework and has the feature that it never fails; moreover, when the chase procedure terminates, it produces a universal solution, which in turn can be used to obtain the certain answers to conjunctive queries. We finally discuss challenges arising when the chase does not terminate. 
\end{abstract}

\newcommand{\Name}{\mathit{Name}}
\newcommand{\HPhone}{\mathit{HPhone}}
\newcommand{\JaccardSim}{\mathit{JaccSim}}
\newcommand{\Age}{\mathit{Age}}
\newcommand{\Zip}{\mathit{Zip}}
\newcommand{\SameHouse}{\mathit{SameHouse}}
\newcommand{\Doe}{\textnormal{\textbf{Doe}}}
\newcommand{\PNumberA}{\texttt{358}}
\newcommand{\PNumberB}{\texttt{635}}
\newcommand{\PNumberC}{\texttt{262}}
\newcommand{\wrt}{w.r.t.\ }

\section{Introduction}
\label{sec:introduction}

Entity resolution is the problem of determining whether different data records refer to the same real-world object, such as the same individual or the same organization, and so on~\cite{BGMS*09,PITP21}.
In this paper, we study entity resolution in combination with query answering in the context of knowledge bases (KBs) consisting of ground atoms and rules specified as tuple-generating dependencies (tgds) and equality-generating dependencies (egds).
These rules have been widely investigates in databases and knowledge representation, e.g. in~\cite{BeVa84,CCDL04,FKMP05,BLMS11,CHKK*13,KrMR19}; in particular, they  allow expressing both data transformations  and  classical axioms  used in Description Logic (DL)~\cite{BCMNP07,ACKZ09} and Datalog +/- ontologies~\cite{CaGL09b}.
In addition, egds are employed to express typical entity resolution rules that one may  write in practice, i.e., rules that enforce equality between two entities, as in~\cite{BiCG22}.
The KBs considered here separate out \emph{entities} and \emph{values} and 
involve $n$-ary predicates that denote $n$-ary relations, $n\geq 1$, over entities and/or values from predefined datatypes. 
As is customary for ontologies,  the  \emph{TBox} is the intensional component (i.e., the rules) of a KB, while the \emph{database} of the KB is its extensional component (i.e., the ground atoms).

\begin{figure*}
{\centering
\begin{multicols}{2}

\scalebox{0.9}{
\begin{tikzpicture}[main/.style = {draw=none}, node distance=2.6cm]
\node[main] (1) {$\Doe_1$};
\node[main] (2) [right of=1] {\PNumberA};
\node[main] (3) [left of=1] {\texttt{John Doe}};

\draw[->] (1) -- node[midway, above] {HPhone}  (2);
\draw[->] (1) -- node[midway, above] {Name}  (3);

\node[main] (4) [below of=1, yshift=1.7cm] {$\Doe_2$};
\node[main] (5) [right of=4] {\PNumberB};
\node[main] (6) [left of=4] {\texttt{J. Doe}};

\draw[->] (4) -- node[midway, below] {HPhone}  (5);
\draw[->] (4) -- node[midway, below] {Name}  (6);

\end{tikzpicture}
}
\par
\hspace{0.2cm}\subcaptionbox{Before entity resolution.\label{Fig:Introduction-without}}[18em]

\scalebox{0.9}{
\begin{tikzpicture}[main/.style = {draw=none}, node distance=2.9cm]

\tikzset{hfit/.style={rounded rectangle, inner xsep=0pt, fill=#1!30},
           vfit/.style={rounded corners, fill=#1!30}}

\node[main] (1) {$\Doe_1$};
\node[main] (2) [right of=1] {\PNumberA};
\node[main] (3) [left of=1] {\texttt{John Doe}};

\node[main] (4) [below of=1, yshift=1.81cm] {$\Doe_2$};
\node[main] (5) [right of=4] {\PNumberB};
\node[main] (6) [left of=4] {\texttt{J. Doe}};


\begin{pgfonlayer}{background}
  \node[fit=(1)(4), vfit=lightgray,  line width=1.6pt] (7) {};
  \node[fit=(2)(5), vfit=lightgray] (8) {};
  \node[fit=(3)(6), vfit=lightgray] (9) {};
\end{pgfonlayer}

\draw[->] (7) -- node[midway, above] {HPhone}  (8);
\draw[->] (7) -- node[midway, above] {Name}  (9);

\end{tikzpicture}

}
\par
\hspace{1.6cm}\subcaptionbox{After entity resolution.\label{Fig:Introduction-with}}[12em]
    
\end{multicols}

\caption{Ground Atoms of the Knowledge Base}\label{Fig:ex-introduction}
}
\end{figure*}
As an example,  Fig.\ \ref{Fig:Introduction-without} depicts a set of ground atoms, where $\Doe_1$ and $\Doe_2$ are entities 
while the rest are values (indicating names and landline home phone numbers). 
Fig.~\ref{Fig:ex-introduction-rules} illustrates a small TBox (containing only egds, for simplicity).
\begin{figure}[h]
	\small
	\begin{align*}
		\arraycolsep=3.0pt
		\begin{array}{ll}
			(s_1)  &  \Name(p_1, n_1) \wedge \Name(p_2, n_2) \wedge  \JaccardSim(n_1, n_2, 0.6)\\ 
			& \rightarrow p_1 = p_2 \\
			(s_2)  &  \Name (p, n_1) \wedge \Name (p, n_2)  \rightarrow n_1 = n_2 \\
			(s_3)  &  \HPhone (p, f_1) \wedge \HPhone (p, f_2) \rightarrow f_1 = f_2
		\end{array}
	\end{align*}
	\normalsize
	\caption{Entity-egds and value-egds.}
	\label{Fig:ex-introduction-rules}
\end{figure}
%
%
To capture entity resolution rules, we allow egds to contain atoms involving 
built-in predicates, such as $\JaccardSim$. In the example, rule $s_1$ states that two names (i.e., strings) with Jaccard similarity above $0.6$ must belong to the same individual. Rules $s_2$ and $s_3$ stipulate  that an individual has at most one name and at most one landline home phone number, respectively. We call \emph{entity-egds} rules that impose equality on two entities (e.g., $s_1$), and we call \emph{value-egds} rules that impose equality on two values (e.g., $s_2$ and $s_3$).

It is now easy  to see that, according to the standard semantics, the database in Fig.~\ref{Fig:Introduction-without} does not satisfy the entity-egd $s_1$ 
(note that the Jaccard similarity of \texttt{John Doe} and \texttt{J. Doe} is $0.625$). Thus, the main challenge one has in practice is to come up with a consistent way to complete or modify the original KB, while respecting all its rules. 

In this paper, we develop a  framework for entity resolution and query answering in KBs, where   the valid models,  called \emph{KB solutions},  must satisfy all entity resolution rules along with all  other KB rules; furthermore, the solutions must include all original data 
(i.e., no information is ever dropped or altered). Our approach is guided by the intuitive principle that for each real-world object there must be a single 
node in the solution that represents all  ``equivalent'' entities denoting the object.
To achieve this, we use  equivalence classes of entities, which become first-class citizens in our framework.
In addition, we relax 
%
 the standard way in which value-egds are satisfied by allowing solutions to use sets of values, thus collecting together all possible values for a given attribute (e.g., the second argument of $\Name$ or $\HPhone$). 
Continuing with the above example, Fig.~\ref{Fig:Introduction-with} shows a new 
set of ground atoms that uses equivalence classes and sets of values.  

We remark that the use  of equivalence classes of entities and sets of values in KB solutions requires a drastic revision of  the classical notion of satisfaction for tgds and egds. 
Intuitively, we interpret egds as matching dependencies~\cite{BeKL13,Fan08}. 
That is, when the conditions expressed by (the body of) an entity-egd or a value-egd hold in the data (and thus two entities or two  values must  be made equal),
we combine them through a merging function. We adopt a general and common merging function that takes the union of the entities or the union of the values. This function actually belongs to the \emph{Union class of match and merge functions} analyzed in~\cite{BGMS*09,BeKL13}.
In this way, we group all possible alternatives for denoting an object into a unique set of entities. Similarly, when two 
values exist where only one value is instead allowed according to the TBox, we explicitly form their union.
Note that  we consider the union of entities a ``global'' feature for a solution, that is, an entity $e$ may belong to exactly one equivalence class;  
in contrast, the union of values is ``local'' to the context in which a value occurs, that is, a particular value may belong to more than one set of values. 
\ifthenelse{\equal{\showproofs}{1}}{For instance, let us add $\HPhone(\Doe_3,\PNumberA)$ to the KB of our example. Then, in a solution satisfying the KB, $\PNumberA$ belongs to two different sets: 
the one shown in Fig.~\ref{Fig:Introduction-with} (i.e., $\{\PNumberA,\PNumberB\}$) and the singleton $\{\PNumberA\}$, which is the set of landline home phone numbers of $\Doe_3$.}

It is worth noting that the semantics we adopt for value-egds allows us to always have a solution (that is, a model) for a KB, even when there are no models according to the standard first-order semantics; for example,
if a value-egd enforces the equality of two different values, then first-order logic concludes that the KB is inconsistent. 
We thus may say that our semantics is inconsistency tolerant with respect to violations of value-egds. Indeed, we collect together all  possible alternative values for individual attributes, whereas in data cleaning the task is to choose one of the alternatives.

We also point out that value-egds may substantially affect entity resolution. 
As an example, consider the database in Fig.~\ref{Fig:Introduction-without} and on top of it the new set of rules given below
\ifthenelse{\equal{\showproofs}{0}}{
\[	
 \small
	\begin{aligned}
		\arraycolsep=3.0pt
		\begin{array}{ll}
			(s'_1) & \Name (p_1, n_1) \wedge \HPhone (p_1, f) \wedge \Name (p_2, n_2) \wedge\\ 
			& \HPhone (p_2, f) \wedge \JaccardSim(n_1, n_2, 0.6) \rightarrow p_1 = p_2 \\
			(s'_2)  &  \Name(p_1, n_1) \wedge \Name (p_2, n_2) \wedge  
			\HPhone (p_1, f_1) \wedge\\ 
			& \HPhone (p_2, f_2)  \wedge \JaccardSim(n_1, n_2, 0.6) \rightarrow f_1 = f_2 
		\end{array}
	\end{aligned}
\]
 \normalsize
}

\ifthenelse{\equal{\showproofs}{1}}{

\[	
 \small
	\begin{aligned}
		\arraycolsep=3.0pt
		\begin{array}{ll}
			(s'_1) & \Name (p_1, n_1) \wedge \HPhone (p_1, f) \wedge \Name (p_2, n_2) \wedge\\ 
			& \HPhone (p_2, f) \wedge \JaccardSim(n_1, n_2, 0.6) \rightarrow p_1 = p_2 \\
			(s'_2)  &  \Name(p_1, n_1) \wedge \Name (p_2, n_2) \wedge  
			\HPhone (p_1, f_1) \wedge\\ 
			& \HPhone (p_2, f_2)  \wedge \JaccardSim(n_1, n_2, 0.6) \rightarrow f_1 = f_2\\[3mm] 
		\end{array}
	\end{aligned}
\]
 \normalsize
 
 }

Rule $s'_1$ states that two entities having similar names and the same landline home phone number denote the same real-world individual, while rule $s'_2$ says that two entities with similar names must have the same phone number. It is easy to see that a TBox having only rule $s'_1$ would not lead to inferring that $\Doe_1$ and $\Doe_2$ denote the same individual. 
However, by virtue of rule $s'_2$, a solution has to group together the two telephone numbers into the same set, thus saying that $\Doe_1$ and $\Doe_2$ have both the landline home phone numbers $\{\PNumberB,\PNumberA\}$. This union ``fires'' rule $s'_1$, and thus entities are resolved, i.e., are put together into the same equivalence class of the solution. 

We finally note that, in order  to maximize entity resolution, our semantics allows the body a rule to  be satisfied by assignments in which different occurrences of the same value-variable are replaced by different sets of values, as long as these sets have a non-empty intersection (instead  of requiring that  all  occurrences are replaced by the same set of values).
%
\ifthenelse{\equal{\showproofs}{0}}{For instance, if we add to the rules in Fig.~\ref{Fig:ex-introduction-rules} the tgd
\small
\[
\begin{array}{l}
	\HPhone(p_1,f) \wedge \HPhone(p_2,f)  \rightarrow \SameHouse(p_1, p_2,f) 
	\end{array}
\]
\normalsize
stating that two entities with the same home phone number live in the same house, and to the database in Fig.~\ref{Fig:Introduction-without} the atom $\HPhone(\Doe_3,\PNumberA)$, we can conclude that the individual denoted by $\Doe_3$ and the individual denoted by $\Doe_1$  and $\Doe_2$ live in the same house (since they share one phone number).}

\ifthenelse{\equal{\showproofs}{1}}{For instance, if we add to the rules in Fig.~\ref{Fig:ex-introduction-rules} the tgd

\small
\[
\begin{array}{l}
	\HPhone(p_1,f) \wedge \HPhone(p_2,f)  \rightarrow \SameHouse(p_1, p_2,f)\\[3mm] 
	\end{array}
\]
\normalsize

\noindent stating that two entities with the same home phone number live in the same house, and to the database in Fig.~\ref{Fig:Introduction-without} the atom $\HPhone(\Doe_3,\PNumberA)$, we can conclude that the individual denoted by $\Doe_3$ and the individual denoted by $\Doe_1$  and $\Doe_2$ live in the same house (since they share one phone number).}
Consistently with this choice, we also require tgds to ``propagate'' such intersections. Then,
through the above tgd, we also infer that the phone number of the house is $\{\PNumberA\}$,which we denote with a fact of the form $\SameHouse([\Doe_1,\Doe_2],[\Doe_3], \{\PNumberA\})$.
Note that the behaviour we described differs from the standard one only for value-variables, since two equivalence classes of entities are always either disjoint or the same class.

In this paper, we formalize the aforementioned ideas and investigate query answering, 
with focus on conjunctive queries (CQs). 
The main contributions are as follows.
\begin{itemize}
    \item We propose a new framework for entity resolution in KBs consisting of tgds and egds, and give rigorous semantics.
    \item 
    We define \emph{universal solutions} and show that, as for standard tgd and egd semantics, universal solutions can be used to obtain the certain answers of CQs. 
    \item We propose a variant of the classical \emph{chase} procedure~\cite{BeVa84,FKMP05} tailored to our approach. 
     An important feature of our chase procedure is that it never fails, even in the presence of egds. At the same time, as in other frameworks, our chase procedure might not terminate.
    \item We show that, when the chase procedure terminates, it returns a universal solution (and thus we have an  algorithm for computing the certain answers to  conjunctive queries).
    \item When the chase procedure does not terminate, defining the result of the chase becomes a more intricate task. 
    \commentout{According to~\cite{BeVa84}, a possible approach is to include in the result of the chase all facts that become persistent from a certain point on of the chase procedure. We point out that the strategy proposed in~\cite{BeVa84}
    does not work in our framework and thus a different notion for the result of the chase is needed, which we however leave for future study.}
    We point out that the strategy proposed in~\cite{BeVa84} does not work in our framework. Thus a different notion of the result of the chase is needed, which we  leave for future study.
 \end{itemize}

\ifthenelse{\equal{\showproofs}{1}}{

\medskip

The rest of the paper is organized as follows. In Section~\ref{sec:preliminaries} we give some background notions useful for the following treatment. In Section~\ref{sec:framework} we propose our framework and give the notions of instance, solution and universal solution for a KB. In Section~\ref{sec:query-answering} we define query answering and show that for CQs it can be solved through query evaluation over a universal solution. In Section~\ref{sec:computing-Univ-Model} we first define the chase procedure in our setting, and discuss the case of terminating chase sequences, for which the chase procedure returns a universal solution; then we attack the challenging case of non-terminating chase sequences, for which we show that previously proposed approaches cannot directly adapted to our setting. In Section~\ref{sec:related} we discuss some related work, and in Section~\ref{sec:conclusions} we conclude the paper.
}

\ifthenelse{\equal{\showproofs}{0}}{
    \noindent Due to space limitations, all proofs have been moved to an appendix, which we provide as supplementary material.
}


\newcommand{\occ}{\mathit{Occ}}

\renewcommand{\vec}[1]{\mathbf{#1}}

\section{Basic Notions}
\label{sec:preliminaries}

\ifthenelse{\equal{\showproofs}{1}}{
\noindent \textbf{Equivalence relations.}
If $\Z$ is a set, then an \emph{equivalence relation on $\Z$} is a binary relation 
 $\theta \subseteq \Z \times \Z$ that is reflexive (i.e., $x \: \theta \: x$, for every $x \in \Z$), symmetric (i.e., $x \: \theta \: y$ implies $y \: \theta \: x$, for every $x,y \in \Z$), and transitive (i.e., $x \: \theta \: y$ and $ y \: \theta \: z$ imply $x \: \theta \: z$, for every $x,y,z \in \Z$).

If $x \in \Z$, then the \emph{equivalence class of $x$ w.r.t.\  $\theta$} is the set $[x]_\theta = \{ y \in \Z \: \mid \: y \: \theta \: x \: \}$. Sometimes, we will write $[x]$ instead of $[x]_\theta$, if $\theta$ is clear from the context. Moreover, we may write, e.g., $[a,b,c]_\theta$ (or simply $[a,b,c]$) to denote the equivalence class consisting of the elements $a$, $b$, and $c$.

The equivalence classes of an equivalence relation $\theta$ have the following well known properties: (i) each equivalence class is non-empty; (ii) for all $x,y\in \Z$, we have that $x\theta y$ if and only if $[x]_\theta=[y]_\theta$; (iii) two equivalence classes are either equal sets or  disjoint sets. Consequently, the equivalence classes of $\theta$ form a partition of $\Z$. The \emph{quotient set $\Z/\theta$ of $\theta$ on $\Z$}  is  the set of all equivalence classes over $\Z$ w.r.t.\ $\theta$, i.e.,  $\Z/\theta= \{[x]_\theta \: | \: x \in \Z\}$. Obviously, given a quotient set $\Z/\theta$, we have that $\Z =  \underset{[x]_\theta \in \Z/\theta}{\bigcup}[x]_\theta$.}

\ifthenelse{\equal{\showproofs}{0}}{We take for granted the notions of equivalence relation and equivalence class. If $\Z$ is a set, $\theta$ is an equivalence relation on $\Z$, and $x \in \Z$, then the \emph{equivalence class of $x$ \wrt $\theta$} is denoted by $[x]_\theta$ (or simply $[x]$, if $\theta$ is clear from the context). Sometimes, we write, e.g., $[a,b,c]_\theta$ (or $[a,b,c]$) to denote the equivalence class consisting of the elements $a$, $b$, and $c$. The \emph{quotient set $\Z/\theta$ of $\theta$ on $\Z$} is  the set of all equivalence classes over $\Z$ w.r.t.\ $\theta$.}

\ifthenelse{\equal{\showproofs}{1}}{

\medskip

\noindent \textbf{Alphabets and conjunctions of atoms.} }

We consider four pairwise disjoint alphabets $\alphaPredicates$, $\alphaEntities$, $\alphaValues$, and $\alphaVariables$. The set $\alphaPredicates$ is a finite alphabet for predicates; it is partitioned into the sets $\alphaOntology$ and $\alphaBuiltIns$, which are the alphabets for \emph{KB} predicates and \emph{built-in} predicates, respectively. The sets $\alphaEntities$, $\alphaValues$, and $\alphaVariables$ are countable infinite alphabets for \emph{entities}, \emph{values}, and \emph{variables}, respectively. For ease of exposition, we do not distinguish between different data types, 
thus  $\alphaValues$ is a single set containing all possible values. 

The number of  arguments of a predicate 
$P \in \alphaPredicates$
is the \emph{arity} of $P$, denoted with $\arity(P)$. With each $n$-ary predicate $P$ we associate a tuple $\type(P)= \tup{\rho_1,\ldots \rho_n}$,  such that, for each $1\leq i\leq n$,  either $\rho_i=\entity$ or $\rho_i=\val$. This tuple  specifies the \emph{types of the arguments of $P$}, i.e., whether each argument of $P$ ranges over entities ($\entity$) or values ($\val$). We also write $\type(P,i)$ for the \emph{type} $\rho_i$ of the $i$-th argument of $P$.

Note that the built-in predicates from $\alphaBuiltIns$ are special, pre-interpreted predicates, whose arguments range only over values, i.e., if $B$ is an $n$-ary built-in predicate, then  $\type(B)$ is the $n$-ary tuple $\tup{\val,\ldots,\val}$. 
The examples in Sec.~\ref{sec:introduction}
use  the Jaccard similarity $\JaccardSim$ as a built-in predicate. 

An \emph{atom} is an expression $P(t_1,\dots,t_n)$, where $P \in \alphaPredicates $, and each $t_i$ is either a variable from $\alphaVariables$ or a constant, which in turn is either an entity from $\alphaEntities$, if $\type(P,i)=\entity$, or a value from $\alphaValues$, if $\type(P,i)=\val$. When $t_i \in \alphaVariables$, we call it an \emph{entity-variable}  if $\type(P,i)=\entity$, or a \emph{value-variable}, if $\type(P,i)=\val$. A \emph{ground} atom is an atom with no variables.

A \emph{conjunction $\phi(\vec{x})$ of atoms} is an expression $P_1(\vec{t}_1) \land \ldots \land P_m(\vec{t}_m)$,
\commentout{
\begin{align}
	\label{eq:conj}
	\begin{array}{c}
		P_1(\vec{t}_1) \land \ldots \land P_m(\vec{t}_m)
	\end{array}
\end{align}}
where each $P_j(\vec{t}_j)$ is an atom such that each variable in $\vec{t}_j$ is among those in $\vec{x}$.
We also require that every variable $x$ in $\vec{x}$ is either an entity-variable or a value-variable and that, if $x$ occurs in an atom whose predicate is built-in (note that thus $x$ is a value-variable), then there exists some other atom $P_j(\vec{t}_j)$ in $\phi(\vec{x})$ such that $P_j \in \alphaOntology$ and $x$ is  in  $\vec{t}_j$.
If a conjunction contains no built-in predicates, we say that it is \emph{built-in free}; 
if it contains no entities or values, we say that it is \emph{constant-free}.

%
%

\ifthenelse{\equal{\showproofs}{1}}{

\medskip

\noindent \textbf{Remark on notation.} A}
As done in this section, in the formulas appearing throughout the paper, we use $e$ to denote an entity from $\alphaEntities$, we use $v$ to denote  a value from $\alphaValues$, we use $c$ to denote a constant (i.e., $c \in \alphaEntities \cup \alphaValues$), we use
$x$, $y$, $w$, and $z$ to denote variables from $\alphaVariables$, and we use $t$ to denote terms (i.e., $t \in \alphaVariables \cup \alphaEntities \cup \alphaValues$). Typically,  we use  $P$ to denote a predicate from $\alphaPredicates$. All the above symbols may appear with subscripts. 
In the examples, we use self-explanatory symbols, and we write entities in bold font 
and values in true type font.

\renewcommand{\Mod}{\mathit{Sol}}

\newcommand{\Emp}{\mathit{Emp}}
\newcommand{\Pers}{\mathit{Pers}}
\newcommand{\CI}{\mathit{CI}}
\newcommand{\phone}{\mathit{phone}}
\newcommand{\name}{\mathit{name}}
\newcommand{\comp}{\mathit{comp}}
\newcommand{\Yahoo}{\textnormal{\textbf{Yahoo}}}
\newcommand{\IBM}{\textnormal{\textbf{IBM}}}
\newcommand{\director}{\mathit{dir}}
\newcommand{\CEO}{\mathit{CEO}}

\newcommand{\Body}{\mathsf{body}}

\section{Framework}
\label{sec:framework}

In this section, we present the syntax and the semantics of a framework for entity resolution in knowledge bases. 
\commentout{While the syntax involves standard notions used in data exchange and integration, the semantics  is novel, because it entails  merging  entities into sets of entities and   merging  values into sets of values.}


\smallskip

\noindent \textbf{Syntax.} A \emph{knowledge base} (KB) $\K$ is a pair $(\T,\DB)$, consisting of a TBox  $\T$ and a database $\DB$.  
The TBox is a finite set of \emph{tuple-generating dependencies} (tgds) and \emph{equality-generating dependencies} (egds).  
A \emph{tgd} is a formula 
\begin{align}
	\label{eq:tgd}
	\begin{array}{c}
		\forall \vec{x}(\phi(\vec{x}) \rightarrow \exists \vec{y} \psi(\vec{x}, \vec{y})),
	\end{array}
\end{align}
where $\phi(\vec{x})$ and $\psi(\vec{x},\vec{y})$  are conjunctions of atoms, such that $\vec{x}$ and  $\vec{y}$ have no variables in common and $\psi(\vec{x},\vec{y})$ is built-in free and, for simplicity, constant-free.  As in \cite{FKMP05}, we assume that all variables in $\vec{x}$ appear in $\phi(\vec{x})$, but not necessarily in $\psi(\vec{x}, \vec{y})$. We call $\phi(\vec{x})$ the body of the tgd, and $\psi(\vec{x},\vec{y})$ the head of it. 



An \emph{egd} is a formula 
\begin{equation}
	\label{eq:egd}
	\forall \vec{x}(\phi(\vec{x}) \ra  y = z),
\end{equation}
where $\phi(\vec{x})$ is a conjunction of atoms  and $y$ and $z$ are distinct variables occurring in $\phi(\vec{x})$, such that either both $y$ and $z$ are  entity-variables (in which case we have an  \emph{entity-egd}) or both $y$ and $z$ are value-variables (in which case we have a \emph{value-egd}). We call $\phi(\vec{x})$ the body of the egd. For \emph{value-egds}, we require that neither $y$ nor $z$  occur in atoms having a built-in as predicate. This ensures that the meaning of built-ins remains fixed.
We will write  $\Body(r)$ to denote  the body of a tgd  or an egd $r$;  furthermore, we may write $r$ with no quantifiers.

\begin{example}
	\label{ex:TBox}
	Let $\T$ be the TBox consisting of the rules  
 \begin{align*}
        \small
		\arraycolsep=1.5pt
		\begin{array}{ll}
			(r_1) & \CI(p_1,\name_1,\phone_1) \wedge \CI(p_2,\name_2,\phone_2) \wedge\\
		      &   \JaccardSim(\name_1,\name_2,0.6) \rightarrow p_1 = p_2 \\
			(r_2) & \CI(p, \name_1, \phone_1) \wedge \CI(p, \name_2, \phone_2)\\ 
                 &   \rightarrow \name_1 = \name_2 \\
            (r_3) & \CI(p,\name_1,\phone_1)\wedge \CI(p,\name_2,\phone_2)\\ 
                 &  \rightarrow \phone_1=\phone_2 \\
            (r_4) & \CI(p,\name,\phone) \rightarrow \Emp(p,\comp) \wedge \CEO(\comp,\director)\\
            (r_5) & \Emp(p,\comp_1) \wedge \Emp(p,\comp_2) \rightarrow \comp_1 = \comp_2\\
           (r_6) & \CI(p_1,\name_1,\phone) \wedge \CI(p_2,\name_2,\phone)\\ 
                &  \rightarrow \SameHouse(p_1,p_2,\phone)
		\end{array}
        \normalsize
	\end{align*}
		Here, $\type(\CI) = \tup{\entity,\val,\val}$, $\type(\Emp) = \tup{\entity,\entity}$, $\type(\CEO) = \tup{\entity,\entity}$, and $\type(\SameHouse) = \tup{\entity,\entity,\val}$.
The predicates have the meaning suggested by their names. In particular, the predicate $\CI$ maintains contact information of individuals, whereas $\Emp$ associates employees to the companies they work for. 
Each of the  six rules makes an assertion about the predicates. In particular,
rule $r_1$ states that if the Jaccard similarity of two names is higher then $0.6$, then these names are names of the same individual. 
\ifthenelse{\equal{\showproofs}{1}}{The TBox also says that everyone can have only one name ($r_2$) and only one the landline phone number ($r_3$); everyone having contact information is also an employee in some company which in turn has a CEO ($r_4$); everyone can work for only one company ($r_5$); two individuals with the same landline phone number live in the same house ($r_6$). \qed}
\qed	
\end{example}

%
%
%
%

The database $\DB$ of a KB $\K$ is a finite set of ground atoms of the form $P(c_1,\ldots,c_n)$ over the alphabets $\alphaPredicates$, $\alphaEntities$, and $\alphaValues$, where $P \in \alphaPredicates$  and each $c_i$ is  an entity from  $\alphaEntities$, if $\type(P,i)=\entity$, or a value from $\alphaValues$, if $\type(P,i)=\val$.

\begin{example}
	\label{ex:DB}
	Let $\DB$ be the database consisting of the atoms
	\begin{align*}
        \small
		\arraycolsep=1.2pt
        \begin{array}{llll}
			(g_1) & \CI(\Doe_1, \texttt{J. Doe}, \PNumberA) &~~~(g_4) & \Emp(\Doe_2, \Yahoo)  \\
			(g_2) & \CI(\Doe_2, \texttt{John Doe}, \PNumberB) &~~~(g_5) & \Emp(\Doe_3, \IBM) \\
			(g_3) & \CI(\Doe_3, \texttt{Mary Doe}, \PNumberA)  &~~~(g_6) & \CEO(\Yahoo, \Doe_1). 
		\end{array}
	\end{align*}
In particular, the atom $g_1$ of $\DB$ asserts that $\Doe_1$ has name \texttt{J.Doe} and phone number \PNumberA. 
 \ifthenelse{\equal{\showproofs}{1}}{Moreover, the database $\DB$ also specifies that $\Doe_1$ has phone number \PNumberA\ ($g_1$), $\Doe_2$ has name \texttt{John Doe} and phone number \PNumberB\ ($g_2$), $\Doe_3$ has name \texttt{Mary Doe} and phone number \PNumberA\ ($g_3$), $\Doe_2$ is employee of $\Yahoo$ ($g_4$), $\Doe_3$ is employee of $\IBM$ ($g_5$), and the CEO of $\Yahoo$ is $\Doe_1$ ($g_6$). \qed}
 \qed
	\end{example}

To ensure that built-in predicates have the same semantics in every KB, 
we assume that we have a fixed (infinite and countable) set $GB$ of ground atoms of the form  $B(v_1,..v_n)$, where $B$ is in $ \alphaBuiltIns$ and $v_1,..v_n$ are in  $\alphaValues$. Intuitively, $GB$ contains all facts about built-in predicates that hold overall. 
Given a  KB $\K=(\T,\DB)$, we assume that $\DB=\DB_O \cup \DB_B$, where $\DB_O$ contains only ground atoms with predicate from $\alphaOntology$ and $\DB_B$ is the (finite) set of all atoms in $GB$ whose built-in predicates and values are mentioned in $\T$ and in $\DB_O$.



\smallskip

\noindent \textbf{Semantics.} Let $\alphaEntityNulls$ and $\alphaValueNulls$ be two
infinite, countable, disjoint sets 
that are also disjoint from  the alphabets introduced in Sec.~\ref{sec:preliminaries}. 
We call $\alphaEntityNulls$ the set of entity-nulls and $\alphaValueNulls$  the set of value-nulls; their union is 
referred to as the set of nulls.
We use
%
$\sig(\K)$ to denote the signature of a KB $\K$, i.e., the set of symbols of $\alphaPredicates$, $\alphaEntities$ and $\alphaValues$ occurring in $\K$. 


The semantics of a KB is given using 
 special databases, called KB \emph{instances}, whose ground atoms have  components that are either equivalence classes of entities and entity-nulls or non-empty sets of values and value-nulls. 

\begin{definition}
\label{def:instance}
Let $\K$ be a KB, $\mathcal{S}$  a subset of  $(\alphaEntities \cap \sig(\K)) \cup \alphaEntityNulls$, and $\sim$  an  equivalence relation on $\mathcal{S}$, called the equivalence relation \emph{associated} with $\K$. An  \emph{instance $\I$ for $\K$ \wrt $\sim$} is a set of \emph{facts} $P(\genericSet_1,\ldots,\genericSet_n)$ such that $P \in \alphaPredicates \cap \sig(\K)$, $\arity(P)=n$, and, for each $1 \leq i \leq n$, we have that $T_i \neq \emptyset$ and either $T_i \in {\mathcal{S}/\hspace*{-3pt}\sim}$, if $\type(P,i)=\entity$, or $T_i \subseteq (\alphaValues \cap \sig(\K)) \cup \alphaValueNulls$, if $\type(P,i)=\val$. 
\end{definition}

To denote an equivalence class  in $\I$, we may use the symbol $E$\ifthenelse{\equal{\showproofs}{1}}{ (possibly with a subscript).}. 
Note that an equivalence class $E$  in $\I$ may contain both entities and entity-nulls. In the sequel, we  may simply refer to $E$ as an equivalence class of entities. Similarly, a non-empty subset of $(\alphaValues \cap \sig(\K)) \cup \alphaValueNulls$  in $\I$ can be simply referred to as a set of values, denoted with the symbol $V$\ifthenelse{\equal{\showproofs}{1}}{ (possibly with a subscript).}. We will use $\genericSet$ (possibly with a subscript) to denote a set that can be either an equivalence class of entities or a set of values.

\begin{definition}
	Let  $\K$ be a KB and let  $\I$ be  an instance for $\K$ \wrt to an equivalence relation $\sim$.

The \emph{active domain of $\I$}, denoted by $\activeD(\I)$,  is the set
	
\centerline{$ \{\genericSet \mid \mbox{there are $P(\genericSet_1,\ldots,\genericSet_n) \in \I$ and  $i\leq n$ with  $T=T_i$}\}$.} 
 
We write  $\activeD_E(\I)$ and $\activeD_V(\I)$ to  denote the set of all 
equivalence classes of entities and the set of all sets of values contained in $\activeD(\I)$, respectively. Obviously, $\activeD(\I)=\activeD_E(\I) \cup \activeD_V(\I)$.
	
		The \emph{underlying domain of $\I$}, denoted by $\underD(\I)$, is the set 
  $\underD_E(\I) \cup \underD_V(\I)$, where

\centerline{$\underD_E(\I)= \{e \mid\mbox{there is $\entitySet \in \activeD_E(\I)$ and  $e \in \entitySet$}\}$}

\noindent and 

\centerline{
$\underD_V(\I)= \{v \mid \mbox{there is $\valueSet \in \activeD_V(\I)$  and  $v \in \valueSet$}\}$.}
\end{definition}

If $\I$ is an instance of $\K$ \wrt  $\sim$,
then it is easy to see that $\I$ is an instance of $\K$ \wrt the equivalence relation $\sim'$ induced on $\underD_E(\I)$ by $\sim$. In what follows, we will consider only instances \wrt equivalence relations over $\S=\underD_E(\I)$.
\commentout{
It is not difficult to see that if $\I$ is an instance of $\K$ \wrt an equivalence relation $\sim$ over a set $\S \subseteq (\alphaEntities \cap \sig(\K)) \cup \alphaEntityNulls$, then $\I$ is an instance of $\K$ \wrt the equivalence relation $\sim'$ over $\underD_E(\I)$ such that $\sim'=\sim \setminus \{\tup{n_i,n_j}~|~n_i,n_j  \not \in \underD_E(\I)\}$. Moreover $\S/\sim \supseteq \underD_E(\I) /\sim'$ and $\sim'$ is unique, i.e., there is no $\sim'' \neq \sim$ such that $\sim''$ is an equivalence relation over $\underD_E(\I)$ and $\I$ is an instance for $\K$ \wrt $\sim''$. Therefore, in the following we will consider only instances \wrt equivalence relations over $\S=\underD_E(\I)$.}
Furthermore, we may simply call $\I$ an instance for $\K$ and leave the equivalence relation associated with $\I$ implicit. 

\begin{example}
	\label{ex:interpretation}
Let $\K=(\T,\DB)$ be a KB such that $\T$ and $\DB$ are as in Ex.~\ref{ex:TBox} and in Ex.~\ref{ex:DB}, respectively. 
%
%
%
%
Then, consider 
the following set $\I$ of facts:
\small
\[
\arraycolsep=1pt
\begin{array}{ll}
\hspace*{-1pt}(d_1)\CI([\Doe_1], \{ \texttt{J. Doe} \}, \{ \PNumberA \} ) & (d_4)\Emp([\Doe_2], [\Yahoo])\\ 
\hspace*{-1pt}(d_2)\CI([\Doe_2], \{ \texttt{John Doe} \}, \{ \PNumberB \} ) & (d_5)\Emp([\Doe_3], [\IBM])\\
\hspace*{-1pt}(d_3)\CI([\Doe_3], \{ \texttt{Mary Doe} \}, \{ \PNumberA \}) & (d_6)\CEO([\Yahoo],[\Doe_1])
\end{array}
\]
\normalsize
 $\I$ is an instance for $\K$ \wrt the identity relation over the set $\mathcal{S}=\{\Doe_1,\Doe_2,\Doe_3,\Yahoo,\IBM\}$.
%
%
%
%
Further, let $\textnormal{\textbf{e}}_1^\bot$ and $\textnormal{\textbf{e}}_2^\bot$ be entity-nulls, and $\sim'$ be an equivalence relation over $\mathcal{S} \cup \{\textnormal{\textbf{e}}_1^\bot,\textnormal{\textbf{e}}_2^\bot\}$, such that $\Doe_1 \sim' \Doe_2$, $\IBM \sim' \textnormal{\textbf{e}}_1^\bot$ and their symmetric versions are the only equivalences in $\sim'$ different from the identity. The following set $\I'$ of facts  is an instance for $\K$ \wrt $\sim'$.
\[
\small
\begin{array}{l}
\CI([\Doe_1,\Doe_2], \{\texttt{J. Doe},\texttt{John Doe} \}, \{\PNumberA,\PNumberB\} ),\\
\CI([\Doe_3],\{ \texttt{Mary Doe} \},\{ \PNumberA \}),\\
\Emp([\Doe_1,\Doe_2],[\Yahoo]),~~\Emp([\Doe_3],[\IBM,\textnormal{\textbf{e}}_1^\bot])\\
\CEO([\Yahoo],[\Doe_1,\Doe_2]),~~\CEO([\IBM,\textnormal{\textbf{e}}_1^\bot], [\textnormal{\textbf{e}}_2^\bot])\\ 
\SameHouse([\Doe_1,\Doe_2],[\Doe_3],\{\PNumberA\}),\\\SameHouse([\Doe_3],[\Doe_1,\Doe_2],\{\PNumberA\}).
\normalsize \hfill \qed
\end{array}
\]
\end{example}

%
%
%
%
%
%
%

To define the notions of satisfaction of tgds and egds by an instance, we first introduce the notion of an \emph{assignment} from a conjunction $\phi(\vec{x})$ of atoms to an instance $\I$ of a KB $\K$. To formalize this notion, we  need a preliminary transformation of $\phi(\vec{x})$ that substitutes each occurrence of a value-variable in $\phi(\vec{x})$ with a fresh variable. We call such fresh variables \emph{set-variables} and denote the result of the transformation $\trasf(\phi(\vec{x}))$. If $x$ is a value-variable in $\vec{x}$, we write $\setvar(x,\trasf(\phi(\vec{x})))$ to denote the set of fresh variables used in $\trasf(\phi(\vec{x}))$ to replace the occurrences of $x$ in $\phi(\vec{x})$.
\commentout{To formalize this notion, we  describe a preliminary transformation of $\phi(\vec{x})$ that substitutes each occurrence of a value-variable in $\phi(\vec{x})$ with a fresh variable. We call such fresh variables \emph{set-variables}. More precisely, we write $\trasf(\phi(\vec{x}))$ to denote a conjunction of atoms obtained as follows:
\begin{tabbing}
	for each\= ~value-variable $x$ in $\vec{x}$\\
	\> for each\= ~occurrence of $x$ in $\phi(\vec{x})$:\\
	\>\> replace $x$ with a fresh symbol $S$ 
\end{tabbing}

If $x$ is a value-variable in $\vec{x}$, we write $\setvar(x,\trasf(\phi(\vec{x})))$ to denote the set of fresh variables used in $\trasf(\phi(\vec{x}))$ to replace the occurrences of $x$ in $\phi(\vec{x})$.
}
For example, if $\phi(\vec{x}) =P_1(x,y,z)\land P_2(y,z) \land P_3(x,w)$, where $\type(P_1)=\tup{\entity,\entity,\val}$, $\type(P_2)=\tup{\entity,\val}$, and $\type(P_3)=\tup{\entity,\val}$, then $\trasf(\phi(\vec{x}))=P_1(x,y,S_1^z)\land P_2(y,S_2^z)\land P_3(x,S_1^w)$, where  $S_1^z$, $S_2^z$, $S_1^w$ are the fresh set-variables introduced by $\trasf$. Also, $\setvar(z,\trasf(\phi(\vec{x})))=\{S_1^z,S_2^z\}$ and  $\setvar(w,\trasf(\phi(\vec{x})))=\{S_1^w\}$. 
In what follows, if $x$ is a value-variable, then   each set-variable in $\setvar(x,\trasf(\phi(\vec{x})))$ will have $x$ as a superscript.

We are now ready to formally define assignments. 

\begin{definition}
	\label{def:assignment-revisited}
	Let $\phi(\vec{x})$ be a conjunction of atoms 
	and let $\I$ be an instance for a KB $\K$ \wrt 
 $\sim$. 
 An \emph{assignment from $\phi(\vec{x})$ to $\I$} is a mapping $\mu$ from the variables and values in $\tau(\phi(\vec{x}))$ to $\activeD(\I)$, defined as follows:
	\begin{enumerate}	
		\item $\mu(x)$ is an equivalence class in $\activeD_E(\I)$, for every entity-variable $x$; 
		\item  $\mu(v)=\valueSet$ such that 
	$\valueSet \in \activeD_V(\I)$ and $v \in \valueSet$, for every value $v$; 
		\item  $\mu(S)$ is a set of values in $\activeD_V(\I)$, for every set-variable $S$; 
		\item $\bigcap_{i=1}^{k} \mu(S_i^x) \neq \emptyset$, for every value-variable $x$ such that $\setvar(x,\trasf(\phi(\vec{x})))=\{S_1^x,\ldots,S_k^x\}$; 
	%
	%
%
%
		\item $\I$ contains a fact of the form $P(\mu(x),\mu(S),[e]_\sim,\mu(v))$, for each atom $P(x,S,e,v)$ of $\trasf(\phi(\vec{x}))$, where $x$ is an entity-variable, $S$ is a set-variable, 
    $e$ is an entity in $\alphaEntities$ and $v$ is a value in $\alphaValues$ (the definition  generalizes in the obvious way for atoms of different form).
	\end{enumerate}	
 \end{definition}

In words, the above definition says that an assignment maps every entity-variable to an equivalence class of entities (Condition~1), every value to a set of values containing it (Condition~2),  and every occurrence of a value-variable to a set of values (Condition~3), in such a way that multiple occurrences of the same value-variable are mapped to sets with a non-empty intersection (Condition~4). This  captures the intuition that, since predicate arguments ranging over values are interpreted through sets of values, a join between such arguments holds when such sets have a non-empty intersection. Finally, Condition~5 states that $\tau(\phi(\vec{x}))$ is ``realized" 
in $\I$.
Note that an assignment also maps values to sets of values (Condition~2). This  allows an assignment to map atoms in  $\phi(\vec{x})$ to facts in $\I$, since predicate arguments ranging over values are instantiated in $\I$ by sets of values.

If $P(\vec{t})$ is an atom of $\phi(\vec{x})$, 
we write $\mu(P(\vec{t}))$ for the fact of $\I$ in  Condition~5 
and call it the \emph{$\mu$-image (in $\I$) 
of $P(\vec{t})$}. We write $\mu(\phi(\vec{x}))$ for the set $\{ \mu(P(\vec{t}))~|~P(\vec{t})\textrm{ occurs in } \phi(\vec{x})\}$, and call it the \emph{$\mu$-image (in $\I$) 
of $\phi(\vec{x})$}.


\begin{example}
	\label{ex:assignment}
Consider  the tgd $r_6$ of Ex.~\ref{ex:TBox} and the instance $\I'$ given in Ex.~\ref{ex:interpretation} and apply 
 $\tau$ to the body of $r_6$ to obtain $\CI(p_1,S^{\name_1},S_1^{\phone}) \wedge CI(p_2,S^{\name_2},S_2^{\phone})$. Let $\mu$ be the following mapping: 
\[
 \small
\begin{array}{l}
	\mu(p_1)=[\Doe_1,\Doe_2],~~\mu(S^{\name_1})=\{\texttt{J. Doe},\texttt{John Doe}\}\\
    \mu(S_1^{\phone})=\{\PNumberA,\PNumberB\},~~\mu(p_2)=[\Doe_3]\\
	\mu(S^{\name_2})=\{\texttt{Mary Doe}\},~~\mu(S_2^{\phone})=\{\PNumberA\}.
 \normalsize
\end{array}
\]
It is easy to see that $\mu$ is an assignment from the body of 
 $r_6$ to $\I'$ (note the non-empty intersection between  $\mu(S_1^{\phone})$ and $\mu(S_2^{\phone})$). Let us now apply $\trasf$ to the head of $r_6$ to obtain $\SameHouse(p_1,p_2,R^{\phone})$. The following mapping $\mu'$ is an assignment from $\SameHouse(p_1,p_2,\phone)$
to $\I'$:
	\[
 \small
\begin{array}{c}
	\mu'(p_1)=[\Doe_1,\Doe_2], \mu'(p_2)=[\Doe_3], \mu'(R^{\phone})=\{\PNumberA\}.
 \normalsize
\end{array}
\]
\end{example}

We are now ready to define the semantics of tgds and egds. 

\begin{definition}
\label{def:rule-satisfaction}
An instance $\I$ for a KB $\K$ satisfies:
\begin{itemize}
\item a tgd of the form (\ref{eq:tgd}), if for each assignment $\mu$ from $\phi(\vec{x})$ to $\I$ there is an assignment $\mu'$ from  $\psi(\vec{x}, \vec{y})$ to $\I$ such that, for each $x$ in $\vec{x}$ occurring in both $\phi(\vec{x})$ and $\psi(\vec{x}, \vec{y})$:  
\begin{itemize}
	\item $\mu(x)=\mu'(x)$, if $x$ is an entity-variable;
	\item $\bigcap_{i=1}^{k} \mu(S_i^x) \subseteq \bigcap_{i=1}^{\ell} \mu'(R_i^x)$, where $\{S_1^x,\ldots,S_k^x\}=\setvar(x,\trasf(\phi(\vec{x})))$ and $\{R_1^x,\ldots,R_\ell^x\}=\setvar(x,\trasf(\psi(\vec{x},\vec{y})))$, if $x$ is a value-variable.
	\end{itemize}
Each such assignment $\mu'$  is called a 
\emph{head-compatible tgd-extension} of $\mu$ to $\I$ (or simply a tgd-extension of $\mu$ to $\I$). 
\item an entity-egd of the form (\ref{eq:egd}), if each assignment $\mu$ from $\phi(\vec{x})$ to $\I$ is such that $\mu(y)=\mu(z)$;
\item a value-egd of the form (\ref{eq:egd}), if each assignment $\mu$ from $\phi(\vec{x})$ to $\I$ is such that
$\mu(S^y)=\mu(S^z)$ for all set-variables $S^y,S^z \in \setvar(y,\trasf(\phi(\vec{x}))) \cup \setvar(z,\trasf(\phi(\vec{x})))$.
\end{itemize} 
\end{definition}
In  words, Def.~\ref{def:rule-satisfaction} expresses the following for each assignment $\mu$ from the body of a rule $r$ to $\I$:

 $(1)$ If $r$ is a tgd, the definition stipulates two different behaviours for frontier variables (i.e., variables occurring in both the body and the head of the tgd). 
 Specifically, for the tgd to be satisfied, it should be possible to extend $\mu$ so that
 (i) every frontier entity-variable is assigned to exactly the same equivalence class both in the body and the head of the rule (this is in line with the standard semantics for frontier variables in tgds); and  (ii) the intersection of the sets of values assigned in the body of the tgd to the various occurrences of a frontier value-variable $x$ is contained in the intersection of the sets of values assigned to the occurrences of $x$ in the head of the tgd. Note that,  while $\mu$ maps all the occurrences of an entity-variable in $\phi(\vec{x})$ to the same  equivalence class, it may map multiple occurrences of a value-variable to different sets of values with non empty-intersection; thus, it is  natural that this intersection is ``propagated" to the head.

$(2)$ If $r$ is an entity-egd, the notion of satisfaction is standard, since it requires that $\mu$ assigns $y$ and $z$ to the same equivalence class of entities. 


$(3)$ If $r$ is a value-egd,  we have again to take into account that each occurrence of $y$ in the rule body can be assigned by $\mu$ to a different set of values (provided that all such sets have a non-empty intersection), and analogously for each occurrence of $z$. The definition stipulates that, for the value-egd to be satisfied, all such sets of values are equal.

\begin{example}
$\I'$ from Ex.~\ref{ex:interpretation} satisfies tgd 
$r_6$ of Ex.~\ref{ex:TBox}, since $\mu'$
is a head-compatible tgd extension of the only assignment $\mu$ from the body of $r_6$ to $\I'$. Indeed, $\mu'(p_1)=\mu(p_1)$, $\mu'(p_2)=\mu(p_2)$, and $\mu'(S^{\phone})=\mu(S_1^{\phone}) \cap \mu(S_2^{\phone})$ (cf.\ Ex.~\ref{ex:assignment}).
It is also easy to see that $\I'$ satisfies all rules of Ex.~\ref{ex:TBox}.~\qed
\end{example}


Finally, we define when an instance is a solution for a KB.

\begin{definition}\label{def:solution}
Let $\K=(\T,\DB)$ be a KB and $\I$  an instance.
\begin{itemize}
    \item 
$\I$ \emph{satisfies}  $\T$ if $\I$ satisfies every tgd and egd in $\T$.

\item  $\I$ satisfies a ground atom $P(c_1,\ldots,c_n)$ in $\DB$ if there is a fact  $P(\genericSet_1,\ldots,\genericSet_n)$ in $\I$ such that $c_i \in \genericSet_i$, for $1\leq i \leq n$.

\item  $\I$ satisfies  $\DB$ if $\I$ satisfies all ground atoms in $\DB$. 


\item  $\I$ is a \emph{solution} for    $\K$, denoted by $\I \models \K$, if $\I$ satisfies $\T$ and $\DB$. The set of all solutions of a KB $\K$ is denoted by $\Mod(\K)$, i.e., $\Mod(\K) = \{ \I~|~\I \models \K \}$.
\end{itemize}
\end{definition}


\begin{example}
\label{exa:univ-sol}
Consider the instances $\I$ and $\I'$ given in Ex.~\ref{ex:interpretation} for the $\K$ of Ex.~\ref{ex:TBox} and Ex.~\ref{ex:DB}. It is easy to see that $\I$ is not a solution for $\K$, whereas $\I'$ is a solution for $\K$.
\qed
\end{example}


\noindent \textbf{Universal solutions.} 
It is well known that the \emph{universal} solutions exhibit good properties that make them to be the preferred solutions (see, e.g., 
\cite{FKMP05,CaGK13,CDLLR07}). 
To introduce the notion of a universal solution in  our framework, we first need to adapt the notion of homomorphism.
We begin with some auxiliary definitions and notation.

\begin{definition}  \label{def:contain} 	Let $\vec{\genericSet}=\tup{\genericSet_1,\ldots,\genericSet_n}$ and
$\vec{\genericSet'}=\tup{\genericSet'_1,\ldots,\genericSet'_n}$ be two tuples such that each $\genericSet_i$ and each $\genericSet'_i$ is either a set of entities and entity-nulls or a set of values and value-nulls.
	
	\begin{itemize} 
		\item $\vec{\genericSet'}$ \emph{dominates} $\vec{\genericSet}$, 
  denoted $\vec{\genericSet} \leq \vec{\genericSet'}$, if $\genericSet_i \subseteq \genericSet'_i$,  for all $i$.
  \item 
  $\vec{\genericSet'}$ \emph{strictly dominates} $\vec{\genericSet}$, denoted   $\vec{\genericSet} < \vec{\genericSet}'$,  
  if  $\vec{\genericSet} \leq \vec{\genericSet'}$ and $\vec{\genericSet} \neq \vec{\genericSet'}$.
		
		\item Let  $P(\vec{\genericSet})$ and $P(\vec{\genericSet'})$ be  facts. 
  $P(\vec{\genericSet'})$ \emph{dominates} $P(\vec{\genericSet})$, 
  denoted $P(\vec{\genericSet}) \leq P(\vec{\genericSet'})$, if $\vec{\genericSet} \leq \vec{\genericSet'}$; 
   $P(\vec{\genericSet'})$ \emph{strictly dominates} $P(\vec{\genericSet})$, denoted $P(\vec{\genericSet}) < P(\vec{\genericSet'})$, if $\vec{\genericSet} < \vec{\genericSet'}$.
	\end{itemize}
	
\end{definition}

\begin{definition}
	\label{def:homomorphism}
	Let $\I_1$ and $\I_2$  be two instances of a KB $\K$.
	A homomorphism $h:\I_1 \rightarrow \I_2$ is a mapping from the elements of $\underD(\I_1)$ to elements of $\underD(\I_2)$ such that:
	\begin{enumerate}
		\item $h(e)=e$, for every entity $e$ in $\underD_E(\I_1) \cap \alphaEntities$;
		\item $h(e_\bot)$ belongs to $\underD_E(\I_2)$, for every entity-null $e_\bot$ in  $\underD_E(\I_1) \cap \alphaEntityNulls$;
		\item $h(v)=v$, for every value $v$ in $\underD_V(\I_1) \cap \alphaValues$;
		\item $h(v_\bot)$ belongs to   $\underD_V(\I_2)$, for every value-null $v_\bot$ in $\underD_V(\I_1) \cap \alphaValueNulls$;
		\item for every $P(\genericSet_1,\ldots,\genericSet_n)$ in $\I_1$, there is a  $P(\genericSettwo_1,\ldots,\genericSettwo_n)$ in $\I_2$ such that $P(h(\genericSet_1),\ldots,h(\genericSet_n))\leq P(\genericSettwo_1,\ldots,\genericSettwo_n)$, where $h(\genericSet_i)=\{h(x) \mid x \in \genericSet_i)\}$, for $1 \leq i \leq n$.
	\end{enumerate}
\end{definition}

In the sequel, we may \ifthenelse{\equal{\showproofs}{1}}{write $h(P(\genericSet_1,\ldots,\genericSet_n))$ instead of $P(h(\genericSet_1),\ldots,h(\genericSet_n))$, and may u} 
use $h(\tup{\genericSet_1,\ldots,\genericSet_n})$ to denote the tuple $\tup{h(\genericSet_1),\ldots,h(\genericSet_n)}$.

\commentout{It is easy to see that the composition of two homomorphisms is  a homomorphism.}

We now define the key notion of a universal solution. 

\begin{definition}
	\label{def:universal}
	A solution $\U$ for a KB $\K$ is \emph{universal} if, for every $\I \in \Mod(\K)$, there is a homomorphism $h: \U \rightarrow \I$.
\end{definition}

The instance $\I'$ in Ex.~\ref{ex:interpretation} is a universal solution for $\K$, as is the instance obtained by eliminating $\textnormal{\textbf{e}}_1^\bot$ from $\I'$.

Two instances $\I_1$ and $\I_2$ are \emph{homomorphically equivalent} if there are homomorphisms $h:\I_1 \rightarrow \I_2$ and  $h':\I_2 \rightarrow \I_1$.
All universal solutions
are homomorphically equivalent.

\section{Query answering}
\label{sec:query-answering}
A conjunctive query (CQ) $q$ is  
a formula $\exists \vec{y}\phi(\vec{x},\vec{y})$,
where 
$\phi(\vec{x}, \vec{y})$ is a built-in free conjunction of atoms. 
The \emph{arity} of $q$ is the number 
of its free variables in $\vec{x}$; we will often write $q(\vec{x})$, instead of just $q$, to indicate the free variables of $q$.

Let  $q(\vec{x}): \exists \vec{x} \phi(\vec{x},\vec{y})$ be a CQ, where $\vec{x}=x_1,\ldots,x_n$. 
Given a KB $\K$ and an instance $\I$ for $\K$, 
the \emph{answer to $q$ on $\I$}, denoted by $q^\I$, is the set of all tuples 
$\tup{\genericSet_1,\ldots,\genericSet_n}$ such that there is an assignment $\mu$ from $\phi(\vec{x}, \vec{y})$ to $\I$ for which 
\begin{itemize}
\item $\genericSet_i = \mu(x_i)$, if $x_i$ is an entity-variable;
\item $\genericSet_i=\bigcap_{j=1}^{k}\mu(S_j^{x_i})$,  if $x_i$ is a value-variable, where $\{S_1^{x_i},\ldots,S_k^{x_i}\}=\setvar(x_i,\trasf(\phi(\vec{x},\vec{y})))$.
\end{itemize}
We will also say that   $\mu$ is an assignment from $q(\vec{x})$ to $\I$.

\begin{example}
\label{ex:query}
Let $q(x)$ be the CQ
$$
\exists p_1,p_2~\CI(p_1,\texttt{J. Doe},x) \wedge \CI(p_2,\texttt{Mary Doe},x)$$
 asking for the phone number in common between (the entities named) \texttt{J. Doe} and \texttt{Mary Doe}.
The answer to $q(x)$ on the instance $\I'$ in Ex.~\ref{ex:interpretation} is
$q^{\I'}=\{\tup{\{\PNumberA\}}\}$. This is obtained through the assignment $\mu_q$
\small
\[
\begin{array}{l}
\hspace*{-0.2cm}\mu_q(p_1)=[\Doe_1,\Doe_2],\mu_q(\texttt{J. Doe})=\{\texttt{J. Doe},\texttt{John Doe}\}\\
\hspace*{-0.2cm}\mu_q(S^x_1)=\{\texttt{\PNumberA,\PNumberB}\},\mu_q(p_2)=[\Doe_3]\\
\hspace*{-0.2cm}\mu_q(\texttt{Mary Doe})= \{\texttt{Mary Doe}\},\mu_q(S_2^x)=\{\PNumberA\}
\end{array}
\]
\normalsize
If $q_1(x): \exists z~\CEO(z,x)
	$ is the query asking for the CEOs, then
\commentout{
\begin{equation}
\label{eq:query-CEO}
\small
	\exists z~\CEO(z,x)
 \normalsize
\end{equation}}
 $q_1^{\I'} = \{ \tup{[\Doe_1, \Doe_2]}, \tup{[\textnormal{\textbf{e}}_2^\bot]} \}$. \qed
 \commentout{, where  answers are obtained through assignments $\mu_1$ and $\mu_2$ 
defined as follows:
\[
\small
\begin{array}{l}
	\mu_1(x)=[\Doe_1, \Doe_2],~~\mu_1(z)=[\mathbf{\Yahoo}]\\
	\mu_2(x)=[\textnormal{\textbf{e}}_2^\bot],~~\mu_2(z)=[\mathbf{\IBM},\textnormal{\textbf{e}}_1^\bot]. \hfill \qed
\end{array} 
\normalsize
\]}
\end{example}

The next result tells how  conjunctive queries are preserved under homomorphisms in our framework. 
\commentout{
following proposition follows from the definitions of answer to a query over an instance, 
assignment, and homomorphism between instances. Intuitively, the proposition below clarifies how in our framework CQ answering is preserved under homomorphism.}
%

\begin{restatable}{proposition}{queryunderhomo}
	\label{pro:query-under-homo}
	Let $q$ be a CQ and let $\K$ a KB.
 If $\I_1$, $\I_2$ are two instances for $\K$ and $h:\I_1 \rightarrow \I_2$ is a homomorphism from $\I_1$ to $\I_2$, 
 then, for every $\vec{\genericSet} \in q^{\I_1}$, there is $\vec{\genericSettwo} \in q^{\I_2}$ such that $h(\vec{\genericSet}) \leq \vec{\genericSettwo}$.
\end{restatable}

\ifthenelse{\equal{\showproofs}{1}}{
    \begin{proof}
	
Without loss of generality, let us consider the query $q(x,y)$ defined has $\exists w \, P_1(x,y,e) \land P_2(x,y,w,v)$, where $x$ is an entity-variable, $y$ and $w$ are value-variables, $e$ is an entity from $\alphaEntities$ and $v$ is a value from $\alphaValues$. By applying the $\tau$ operator to (the conjunction of atoms in) $q(x,y)$ we get $P_1(x,S_1^y,e) \land P_2(x,S_2^y,S^w,v)$. Then, by the definition of answer to a query over an instance it follows that for every tuple $\tup{E,V} \in q^{\I_1}$ there is an assignment $\mu$ from $q(x,y)$ to $\I_1$ such that, $\mu(x)=E$, $\mu(S_1^y)=V_1$, $\mu(S_2^y)=V_2$, $\mu(S^w)=V_3$, $\mu(v)=V_4$, $V=V_1 \cap V_2 =\emptyset$, $v \in V_4$, and $\I_1$ contains the facts $P_1(E,V_1,[e]_{\sim^1})$ and $P_2(E,V_2,V_3,V_4)$ (where $\sim^1$ is the equivalence relation associated to $\I_1$). Since $h$ is homomorphism from $\I_1$ to $\I_2$, the instance $\I_2$ has to contain the facts $P_1(\hat{E},\hat{V}_1,\hat{E}_1)$ and $P_2(\hat{E},\hat{V}_2,\hat{V}_3,\hat{V}_4)$, such that $P_1(h(E),h(V_1),h([e]_{\sim^1})) \leq P_1(\hat{E},\hat{V}_1,\hat{E}_1)$ and $P_2(h(E),h(V_2),h(V_3),h(V_4)) \leq P_2(\hat{E},\hat{V}_2,\hat{V}_3,\hat{V}_4)$
We then construct a mapping $\mu'$ from the variables and values in $P_1(x,S_1^y,e) \land P_2(x,S_2^y,S^w,v)$ to $\activeD_E(\I_2)$ as follows: $\mu'(x)=\hat{E}$, $\mu'(S_1^y)=\hat{V}_1$, $\mu'(S_2^y)=\hat{V}_2$, $\mu'(S^w)=\hat{V}_3$, and $\mu'(v)=\hat{V}_4$. It is not difficult to see that $\mu'$ is an assignment from $q(x,y)$ to $\I_2$. In particular: 
\begin{itemize}
\item since $e \in [e]_{\sim^1}$ and $h(e)=e$, then $e \in h([e]_{\sim^1})$, and since $h([e]_{\sim^1}) \subseteq \hat{E}_1$, it follows that $e \in \hat{E}_1$, which implies that $\hat{E}_1=[e]_{\sim^2}$ (where $\sim^2$ is the equivalence relation associated to $\I_2$);
\item since $V_1 \cap V_2 \neq \emptyset$ and since $h(V_1 \cap V_2) \subseteq h(V_1) \cap h(V_2) \subseteq \hat{V}_1 \cap \hat{V}_2$, it follows that that $\hat{V}_1 \cap \hat{V}_2 \neq \emptyset$;
\item since $v \in V_4$ and $h(v)=v$, then $v \in h(V_4)$, and since $h(V_4) \subseteq \hat{V}_4$, it follows that $v \in \hat{V}_4$.
\end{itemize}
The above assignment $\mu'$ testifies that $\tup{\hat{E},\hat{V}} \in q^{\I_2}$, where $\hat{V}=\hat{V}_1 \cap \hat{V}_2$. This proves the thesis, since, as said, $h(E) \subseteq \hat{E}$ and $h(V) \subseteq \hat{V}$.
\end{proof}
}
{}


%

When querying a KB $\K$, we are interested in reasoning over all solutions for $\K$. We adapt  the classical notion of certain answers to our framework. A tuple $\vec{\genericSet} = \tup{\genericSet_1,\ldots,\genericSet_n}$ is  \emph{null-free} if
 each $T_i$ is non-empty and contains no nulls.


\begin{definition}
	\label{def:certain-answer}
		Let $q$ be a CQ and let $\K$ be a KB. 
        A null-free tuple $\vec{T}$ 
        is a \emph{certain answer} to $q$ \wrt $\K$ if
	\begin{enumerate}
		\item  for every solution $\I$ for $\K$, there is a tuple 
        $\vec{T'} \in q^{\I}$ such that $\vec{T} \leq \vec{T'}$;
		\item   there is no null-free tuple $\vec{T'}$ that satisfies $1.$ and  $\vec{T} < \vec{T'}$.
	\end{enumerate}
We write $\cert(q,\K)$ for the set of certain answers to $q$.
\end{definition}

Note that the second condition in the above definition asserts  that a certain answer has to be a maximal null-free tuple with respect to the tuple dominance order $\leq$ given in Def.~\ref{def:contain}, among the tuples satisfying Condition $1.$ 
\ifthenelse{\equal{\showproofs}{1}}{The following example clarifies the importance of this condition.

\begin{example}
	In every solution for the KB $\K$ used in our ongoing example, $\Doe_1$ and $\Doe_2$ are in the same equivalence class denoting a CEO. 
    Then, the null-free tuple $\tup{\{\Doe_1,\Doe_2\}}$ is a certain answer to the query $q_1(x): \exists z \CEO(z,x)$. Observe that  the null-free tuple $\tup{\{\Doe_1\}}$ satisfies the first (but not the second) condition in Def.~\ref{def:certain-answer}. However, we do not want this tuple to be  a certain answer since it is not an equivalence class in any solution.  \qed
\end{example}
}

 \ifthenelse{\equal{\showproofs}{0}}{The next result tells that if two sets of entities appear in a certain answer or in different certain answers, then either they are the same set or they are disjoint. In particular,  the sets of entities that appear in  certain answers can be viewed as equivalence classes of some equivalence relation.}

 \ifthenelse{\equal{\showproofs}{1}}{The next results tell that if two sets of entities appear in a certain answer or in different certain answers, then either they are the same set or they are disjoint. In particular,  the sets of entities that appear in  certain answers can be viewed as equivalence classes of some equivalence relation.}
 
\begin{restatable}{proposition}{partitionofcertainanswers}
	\label{pro:partition-of-certain-answers}
Let $q$ be a CQ of arity $n$, let $\K$ be a KB, and let $\vec{\genericSet}=\langle T_1, \dots , T_n \rangle$ and $\vec{\genericSet'}=\langle T'_1, \dots, T'_n \rangle$ be two certain answers to $q$ \wrt $\K$. If $T_i$ and $T'_j$ are sets of entities, then either $T_i = T'_j$ or $T_i \cap T'_j = \emptyset$, where $1\leq i,j\leq n$.
\end{restatable}

\ifthenelse{\equal{\showproofs}{1}}{
    \begin{proof}
	Towards a contradiction, assume that $T_i$, $T'_j$ are two non-empty sets of entities such that $T_i \neq T'_j$ and $T_i \cap T'_j \neq\emptyset$. Then $T_i \subsetneqq T_i \cup T'_j$ or $T'_j \subsetneqq T_i \cup T'_j$. Let's assume $T_i \subsetneqq T_i \cup T'_j$. Since $\langle T_1, \dots, T_n \rangle$ and $\langle T'_1, \dots, T'_n \rangle$ are certain answers of $q$, for every solution $\I$, there are tuples $\langle L_1, \dots, L_n \rangle$, $\langle L'_1, \dots, L'_n \rangle$ in $q^\I$ such that for every $k$ we have $T_k \subseteq L_k$ and $T'_k \subseteq L'_k$. Thus we have that $T_i \subseteq L_i$ and $T'_j \subseteq L'_j$. Since $T_i \cap T'_j \neq \emptyset$, it follows that $L_i \cap L'_j \neq \emptyset$. However, $L_i$, $L'_j$ are equivalence classes in the solution $\I$, hence they must coincide, i.e. $L_i = L'_j$. It follows that $T_i \cup T'_j \subseteq L_i = L'_j$. Consider now the tuple $\langle T_1, \dots, T_{i-1}, T_i \cup T_j ,T_{i+1}, \dots ,T_n \rangle$, which is obtained from $\langle T_1, \dots, T_n \rangle$ by replacing $T_i$ by $T_i \cup T_j$. Note that $\langle T_1, \dots, T_{i-1}, T_i ,T_{i+1}, \dots ,T_n \rangle < \langle T_1, \dots, T_{i-1}, T_i \cup T_j ,T_{i+1}, \dots ,T_n \rangle$ because $T_i \subsetneqq T_i \cup T'_j$. However for every solution $\I$ of $\K$, we have that $T_i \cup T'_j \subseteq L_i$, and for all $1 \leq k \leq n$ such that $k \neq i$, we have that $T_k \subseteq L_k$. This means that $\langle T_1, \dots, T_{i-1}, T_i \cup T_j ,T_{i+1}, \dots ,T_n \rangle$ satisfies Condition~$(i)$ of Definition~\ref{def:certain-answer}, and therefore $\langle T_1, \dots, T_n \rangle$ does not respect Condition~$(ii)$ of Definition~\ref{def:certain-answer}), and thus it is not a certain answer, which is a contradiction. Note finally that if we assume $T'_j \subsetneqq T_i \cup T'_j$ we contradict that $\langle T'_1, \dots, T'_n \rangle$ is a certain answer by using an analogous argument. 
\end{proof}

}

\ifthenelse{\equal{\showproofs}{1}}{
The following property easily follows from the above proposition.
\begin{restatable}{corollary}{corpartitionofcertainanswers}
	\label{cor:partition-of-certain-answers}
	Let $\K$ be a KB, $q$ be a CQ, and 
	$\langle T_1, \dots ,T_n \rangle$ be 
	a certain answer to $q$ with respect to $\K$. If $T_i$ and $T_j$ are sets of entities 
	such that $i \neq j$, then either $T_i = T_j$ or $T_i \cap T_j = \emptyset$.
\end{restatable}
    \begin{proof}
 It follows from Proposition~\ref{pro:partition-of-certain-answers} when $\vec{T}=\vec{T'}$. 
\end{proof}
}

In data exchange and related areas, the certain answers to CQs can be obtained by evaluating the query on a universal solution and then applying an operator $\downarrow$ that eliminates the tuples that contain nulls (see, e.g.,~\cite{FKMP05,CDLLR07}). 
The operator $\downarrow$ can easily be adapted to our framework as follows.
If $\entitySet$ is a set of entities and entity-nulls,
then  $\entitySet{\downarrow}=\{e~|~e \in \alphaEntities \cap \entitySet\}$. Similarly, if $\valueSet$ is 
 a non-empty set of values and value-nulls, then $\valueSet{\downarrow}=\{v~|~v \in \alphaValues \cap \valueSet\}$. In words,  $\downarrow$ removes all nulls from $\entitySet$ and from $\valueSet$. If $\langle T_1, \dots,T_n \rangle$ is a tuple such that each $T_i$ is either a set of entities and entity-nulls or a set of values and value-nulls, then we set $\langle T_1, \dots,T_n \rangle{\downarrow} = \langle T_1 {\downarrow}, \dots,T_n {\downarrow} \rangle$.
Finally, given a set $\Theta$ of tuples of the above form,
$\Theta {\downarrow}$ is the set obtained from $\Theta$ by removing all tuples in $\Theta$ containing a $T_i$ such that $T_i {\downarrow} = \emptyset$, and replacing every  other tuple $\langle T_1, \dots,T_n \rangle$ in $\Theta$ by the tuple $\langle T_1, \dots,T_n \rangle {\downarrow}$.

The next example shows that, in our framework, the certain answers to a CQ $q$  cannot always be obtained by evaluating $q$ on a universal solution and then applying 
 $\downarrow$.
\commentout{
Whereas under standard semantics, 
the certain answers to a CQ $q$ over theories admitting universal solutions can be often obtained by computing the answers to $q$ over one such solution and then filtering these answers through the $\downarrow$ operator (see, e.g.,~\cite{FKMP05,CDLLR07}), in our framework this is not the case, as shown by the example below.
}

\begin{example}
Let $P_1$ and $P_2$ be
such that $\type(P_1)=\type(P_2)=\tup{\entity,\val}$, $\T'$ be the TBox consisting of the rules
$$
\small
	\begin{array}{l}
	P_1(x,y) \rightarrow P_2(x,y),~~
	P_1(x,y) \land P_1(x,z) \rightarrow y=z
 \normalsize 
\end{array}
$$
and $\DB'=\{P_1(\textnormal{\textbf{e}},1), P_1(\textnormal{\textbf{e}},2)\}$.
Consider  the following universal solutions for $\K'=(\T',\DB')$
$$
\begin{array}{lcl}
\I_1 &=& \{P_1([\textnormal{\textbf{e}}],\{1,2\}), P_2([\textnormal{\textbf{e}}],\{1,2\})\}\\
\I_2 &=& \{P_1([\textnormal{\textbf{e}}],\{1,2\}), P_2([\textnormal{\textbf{e}}],\{1\}), P_2([\textnormal{\textbf{e}}],\{1,2\})\}
\end{array}
$$
and the query $q(x,y): P_2(x,y)$.
Then
$q^{\I_1}{\downarrow} = q^{\I_1} = \{\tup{[\textnormal{\textbf{e}}],\{1,2\}}\}$, and $q^{\I_2}{\downarrow} =  q^{\I_2}=\{\tup{[\textnormal{\textbf{e}}],\{1,2\}}, \tup{[\textnormal{\textbf{e}}],\{1\}}\}$. Thus, $q^{\I_2}{\downarrow}$ does not coincide with the set of certain answers to $q$ w.r.t.\ $\K'$, because the tuple $\tup{[\textnormal{\textbf{e}}],\{1\}}$ does not satisfy the second condition in Def.~\ref{def:certain-answer}.
\qed
	\end{example}


Intuitively, in the above example the universal solution $\I_2$ is not ``minimal'', in the sense that the fact 
$P_2([\textnormal{\textbf{e}}],\{1\})$ in $\I_2$ is dominated by another fact of $\I_2$, namely $P_2([\textnormal{\textbf{e}}],\{1,2\})$. 
This behaviour causes that the answer to the query over $\I_2$ contains also tuples that are not maximal with respect to tuple dominance (cf.\ Def.~\ref{def:certain-answer}).  
\ifthenelse{\equal{\showproofs}{1}}{Note also that  the situation is even  more complicated in the presence of nulls. For example, we might have an instance containing both facts $Q([\textnormal{\textbf{e}}],\{1,2,n\})$ and $Q([\textnormal{\textbf{e}}],\{1,n'\})$, where $n$ and $n'$ are two different value-nulls. In this case, no fact in the solution is dominated by another fact, still ``non-minimality"  shows up if we eliminate the nulls.} 
This suggests that we need some additional processing besides elimination of nulls. 
We modify the  operator $\downarrow$ by performing a further \emph{reduction} step. 

\commentout{
\begin{definition} \label{def:dareduced}
 Given a query $q$ and an instance $\I$ for a KB $\K$, we write $q^\I\dareduced$ to denote the set of null-free tuples obtained through following three steps:
\begin{enumerate}
\item evaluate $q$ on $\I$ and obtain $q^\I$;
\item compute $q^\I\downarrow$;
\item remove from  $q^\I\downarrow$ all tuples strictly dominated by other tuples in the set, that is, if $\vec{\genericSet}$ and $\vec{\genericSet'}$ are two tuples in $q^\I\downarrow$ such that $\vec{\genericSet} < \vec{\genericSet'}$, then remove $\vec{\genericSet}$ from $q^\I\downarrow$.\nb{If needed, this definition can be shortened}
\end{enumerate}
\end{definition}
}

\begin{definition} \label{def:dareduced}
 Given a query $q$ and an instance $\I$ for a KB $\K$, we write $q^\I{\dareduced}$ to denote the set of null-free tuples obtained by removing from $q^\I{\downarrow}$ all tuples strictly dominated by other tuples in the set, that is, if $\vec{\genericSet}$ and $\vec{\genericSet'}$ are two tuples in $q^\I{\downarrow}$ such that $\vec{\genericSet} < \vec{\genericSet'}$, then remove $\vec{\genericSet}$ from $q^\I{\downarrow}$.

\end{definition}

%
%
%
%
%
%

\ifthenelse{\equal{\showproofs}{1}}{
The following property is needed to prove the last theorem of this section. It easily follows from Prop.~\ref{pro:query-under-homo} and Def.~\ref{def:homomorphism}.

\begin{restatable}{corollary}{queryoverunivsolution}
\label{cor:query-over-univ-solution}
Let $\K$ be a KB, $q$ be a CQ, $\U$ be a universal solution for $\K$, and $\vec{\genericSet} \in q^{\U}$. Then, for every $\I \in \Mod(\K)$ there is a tuple $\vec{\genericSet'} \in q^{\I}$ such that  $\vec{\genericSet}\downarrow \leq \vec{\genericSet'}$. 
\end{restatable}

\begin{proof} Obviously, $\vec{\genericSet'}\downarrow \leq \vec{\genericSet'}$.
Since $\U$ is universal, there is a homomorphism $h$ from $\U$ to $\I$. From Proposition~\ref{pro:query-under-homo} it follows that there is $\vec{\genericSet'} \in q^{\I}$ such that $h(\vec{\genericSet}) \leq \vec{\genericSet'}$. From Definition~\ref{def:homomorphism} (homomorphism) it follows that $\vec{\genericSet}\downarrow \leq h(\vec{\genericSet})$, and thus $\vec{\genericSet}\downarrow \leq \vec{\genericSet'}$. 
\end{proof}
}

We are now able to present the main result of this section, which asserts that universal solutions can be used to compute the certain answers to  CQs in our framework.

\begin{restatable}{theorem}{thmqueryanswering}
	\label{ref:thm-query-answering}
	Let $q$ be a CQ, let $\K$ be a KB,  and let $\U$ be a universal solution for $\K$. Then $\cert(q,\K) =  q^\U{{\dareduced}}$.
\end{restatable}

\ifthenelse{\equal{\showproofs}{1}}{
    \begin{proof}
We first show that $q^\U{\dareduced} \subseteq \cert(q,\K)$. Let $\vec{\genericSet} \in {q^{\U}{\dareduced}}$, which of course implies $\vec{\genericSet} \in q^{\U}$. By Corollary~\ref{cor:query-over-univ-solution} we have that for every $\I \in \Mod(\K)$ there is a tuple $\vec{\genericSet'} \in q^{\I}$ such that $\vec{\genericSet}{\downarrow} \leq \vec{\genericSet'}$. Since $\vec{\genericSet} \in {q^{\U}{\dareduced}}$, it is null-free and thus $\vec{\genericSet}{\downarrow}=\vec{\genericSet}$. Hence, Condition~$(i)$ of Definition~\ref{def:certain-answer} (certain answers) is satisfied by $\vec{\genericSet}$. 
Let us now assume by contradiction that $\vec{\genericSet}$ does not satisfy Condition~$(ii)$ of Definition~\ref{def:certain-answer}. This means that there exists a null-free tuple $\vec{\genericSet'}$ satisfying Condition~$(i)$ of Definition~\ref{def:certain-answer} such that $\vec{\genericSet} < \vec{\genericSet'}$. Since $\vec{\genericSet'}$ satisfies Condition~$(i)$ of Definition~\ref{def:certain-answer}, there must exist a tuple $\vec{\genericSettwo} \in q^\U$ such that $\vec{\genericSet'} \leq \vec{\genericSettwo}$, and thus obviously $\vec{\genericSet} < \vec{\genericSettwo}$. We can apply the $\downarrow$ operator to both $\vec{\genericSet}$ and $\vec{\genericSettwo}$ and get $\vec{\genericSet}{\downarrow} < \vec{\genericSettwo}{\downarrow}$, and since, as said, $\vec{\genericSet}=\vec{\genericSet}{\downarrow}$, we have that $\vec{\genericSet}< \vec{\genericSettwo}{\downarrow}$, which is a contradiction. Indeed, $\vec{\genericSet} \in q^{\U}{\dareduced}$ (by hypothesis) and $\vec{\genericSettwo}\downarrow \in q^{\U}{\downarrow}$, and thus it is not possible that $\vec{\genericSettwo}$ dominates $\vec{\genericSet}$ (cf.\ Step~3 of the definition of the $\dareduced$ operator). 

We now show that $\cert(q,\K)\subseteq q^\U{\dareduced}$. Let $\vec{\genericSet} \in \cert(q,\K)$. By Condition~$(i)$ of Definition~\ref{def:certain-answer} there exists a tuple $\vec{\genericSettwo} \in q^\U$ such that $\vec{\genericSet} \leq \vec{\genericSettwo}$, which implies $\vec{\genericSet} \leq \vec{\genericSettwo}{\downarrow}$ (since $\vec{\genericSet}$ is null-free). 
By Corollary~\ref{cor:query-over-univ-solution}, we have that for each $\I \in \Mod(\K)$ there is a tuple $\vec{\genericSet'} \in q^{\I}$ such that $\vec{\genericSettwo}{\downarrow} \leq \vec{\genericSet'}$. This last fact implies that $\vec{\genericSettwo}{\downarrow}$ is a null-free tuple satisfying Condition~$(i)$ of Definition~\ref{def:certain-answer}. Thus, if $\vec{\genericSet} < \vec{\genericSettwo}{\downarrow}$ then $\vec{\genericSet}$ is not a certain answer, which contradicts the assumption. Therefore, it must be $\vec{\genericSet} = \vec{\genericSettwo}{\downarrow}$. This shows that $\vec{\genericSet} \in q^\U{\downarrow}$. Let us assume by contradiction that $\vec{\genericSet} \not \in q^\U{\dareduced}$. By definition of the $\dareduced$ operator, the above assumption implies that there exists a tuple $\vec{X} \in q^\U{\dareduced}$ such that $\vec{T} < \vec{X}$, and since $q^\U{\dareduced} \subseteq \cert(q,\K)$, as shown in the first part of this proof, $\vec{X}$ satisfies Condition~$(i)$ of Definition~\ref{def:certain-answer}, which in turn implies that $\vec{T} \not \in \cert(q,\K)$, which is a contradiction.
\end{proof}
}

\begin{example}
    As seen earlier in Ex.\ref{exa:univ-sol},  $\I'$ is a universal solution for $\K$. 
    Then, for the query $q_1(x): \exists z CEO(z,x)$,  we have that   $\cert(q_1,\K) = q_1^{\I'}{\dareduced}=\{ \tup{[\Doe_1,\Doe_2]}\}$. \qed
\end{example}

\newcommand{\facts}{\mathrm{Facts}}
\newcommand{\factSequence}{\mathrm{FactDerivation}}
\newcommand{\limit}{\mathit{lim}}
\newcommand{\pred}{\mathit{Pred}}
\newcommand{\vt}{\vec{t}}
\newcommand{\vx}{\vec{x}}

\section{Computing a universal model}
\label{sec:computing-Univ-Model}

In this section, we adapt the well-known notion of restricted chase~\cite{BeVa84,JoKl84,FKMP05,CDLLR07}
to our framework. 
\commentout{
Intuitively, given a KB $\K=(\T,\DB) $, to chase $\K$ we define a procedure that starts from an instance for $\K$ ``specular'' to the database $\DB$, which we call the base instance of $\DB$, and incrementally constructs an instance that satisfies the tgds and the egds of the TBox $\T$. This is obtained by iteratively applying three chase rules, one for each kind of rule that may occur in $\T$, as long as they are triggered in the (current) instance.
}
Interestingly, the chase procedure we define never produces a failure, 
unlike, e.g., the restricted chase procedure in the case of standard data exchange, where the application of egds may cause a failure when two different constants have to be made equal. 
Instead, in our framework, when an entity-egd forces two different equivalence classes of entities to be equated, we combine the two equivalence classes into a bigger equivalence class.  
Similarly, when a value-egd forces two different sets of values to be equated, we take the union of the two sets and modify the instance accordingly.  However, it is possible that  the chase procedure may have infinitely many steps, each producing a new instance. As a consequence, some care is required in defining the result of the application of this (potentially infinite) procedure to a KB $\K$, so that we can obtain an instance that can be used (at least in principle) for query answering. We call such instance \emph{the result of the chase of the KB $\K$}, and distinguish the case in which the chase terminates from the case in which it does not. In the former case, the result of the chase of $\K$ is simply the instance produced in the last step of the chase procedure, and we show that this instance is a universal solution for $\K$. In the latter case, we point out that previous approaches from the literature for infinite standard chase sequences under tgds and egds cannot be smoothly adapted to our framework, and we leave it open for this case how to define the result of the chase so that it is a universal solution.


We start with the notion of the \emph{base instance} for a KB. 
Given a KB $\K=(\T,\DB)$, we define the set 
$$\I^{\DB}=\{P(\{c_1\},\ldots,\{c_n\})~|~P(c_1,\ldots,c_n) \in \DB\}.$$ 
 $\I^{\DB}$ is an instance for $\K$ \wrt the identity relation $\id$ over the set $\S=\alphaEntities \cap \sig(\K)$. We call $\I^\DB$ the \emph{base instance} for  $\K$.
Note that the base instance for $\K$ is also a base instance for every KB having $\DB$ as database, and is a solution for $(\emptyset,\DB)$.
%
As an example, note that the instance $\I$ of Ex.~\ref{ex:interpretation} is the base instance for the KB $\K$ defined in Ex.~\ref{ex:TBox} and in Ex.~\ref{ex:DB}.

We next define three chase steps, one for tgds, one for entity-egds, and one for value-egds.

\begin{definition}
	\label{def:chaseStep}	
	Let $\K=(\T,\DB)$ be a KB and  $\I_1$  an instance for $\K$ \wrt an equivalence relation $\sim^1$ on $\underD_E(\I_1)$.
	
	\begin{itemize}
		\item (tgd) Let $r$ be a tgd of the form (\ref{eq:tgd}). Without loss of generality, we assume that all atoms in $\psi(\vec{x},\vec{y})$ are of the form $P(x_1,\ldots,x_k,y_1,\ldots,y_\ell)$, where $x_1,\ldots,x_k$ belong to $\vec{x}$ and $y_1,\ldots,y_\ell$ belong to $\vec{y}$, and we denote with $\vec{y_e}=y_e^1,\ldots y_e^h$ and $\vec{y_v}=y_v^1,\ldots,y_v^j$ the entity-variables and value-variables in $\vec{y}$, respectively. Let $\mu$ be an assignment from $\phi(\vec{x})$ to $\I_1$ such that there is no $\psi(\vec{x},\vec{y})$-compatible tgd-extension of  $\mu$ to $\I_1$. 
		We say that $r$ is \emph{applicable} to $\I_1$ with  $\mu$ (or  that $\mu$ \emph{triggers} $r$ in $\I_1$),	and	construct $\I_2$ via the following procedure:\\[0.2cm]
            \hspace*{0.3cm} \textbf{let} $\{f_e^1,\ldots,f_e^h\} \subseteq \alphaEntityNulls$ and $\{f_v^1,\ldots,f_v^j\} \subseteq \alphaValueNulls$ be\\
            \hspace*{0.5cm}  two sets of fresh nulls (i.e., not occurring in $\I_1$),\\ 
            \hspace*{0.5cm}  which are distinct from each other\\
			\hspace*{0.3cm} \textbf{put} $\I_2:=\I_1$\\
			\hspace*{0.3cm} \textbf{for each} atom $P(x_1,\ldots,x_k,y_1,\ldots,y_\ell)$ in $\psi(\vec{x},\vec{y})$ \textbf{do}\\
			\hspace*{0.6cm} $\I_2:=\I_2 \cup \{P(T_1,\ldots,T_k,U_1,\ldots,U_\ell)\}$\\[2mm]
			where, 
			\begin{itemize}
				\item for $1 \leq i \leq k$, we have
    $T_i=\mu(x_i)$, if $x_i$ is an entity-variable, or $T_i = \bigcap_{p=1}^{m}\mu(S_p^{x_i})$, if $x_i$ is value-variable and $\{S_1^{x_i},\ldots,S_m^{x_i}\}=\setvar(x_i,\trasf(\phi(\vec{x})))$;
				\item for $1 \leq i \leq \ell$, we have
    $U_i=\{f_e^s\}$ if $y_i=y_e^s$, with $1 \leq s \leq h$, or $U_i=\{f_v^q\}$ if $y_i=y_v^q$, with $1 \leq q \leq j$. (Thus, each singleton $\{f_e^s\}$ is  a new equivalence class.)
		  \end{itemize} 
	
		\item (entity-egd) Let $r$ be an entity-egd of the form (\ref{eq:egd}), and 
		 $\mu$  an assignment from $\phi(\vec{x})$ to $\I_1$ such that $\mu(y) \neq\mu(z)$. We say that $r$ is \emph{applicable} to $\I_1$ with  $\mu$ (or  that $\mu$ \emph{triggers} $r$ in $\I_1$), and we construct $\I_2$ from $\I_1$ by replacing 
			in $\I_1$ all occurrences of $\mu(y)$ and  $\mu(z)$ with $\mu(y) \cup \mu(z)$. \\
   (Thus, we merge two equivalence classes into a new one.)
		\item (value-egd) let $r$ be a value-egd of the form~(\ref{eq:egd}). 
        Let $\{S_1^y,\ldots,S_m^y\}=\setvar(y,\trasf(\phi(\vec{x})))$, $\{S_1^z,\ldots,S_k^z\}=\setvar(z,\trasf(\phi(\vec{x})))$, $1\leq i \leq m$, $1\leq j\leq k$, and $\mu$ be an assignment from $\phi(\vec{x})$ to $\I_1$ such that $\mu(S^y_i) \neq\mu(S^z_j)$.
        We say that $r$ is \emph{applicable} to $\I_1$ with  $\mu$ (or  that $\mu$ \emph{triggers} $r$ in $\I_1$), and we construct 
		$\I_2$ from $\I_1$ by replacing in the image  $\mu(\phi(\vec{x}))$ 
		each set $\mu(S_1^y),\ldots,\mu(S_m^y),\mu(S_1^z),\ldots$ and $\mu(S_k^z)$ with $\mu(S_1^y) \cup \ldots \cup \mu(S_m^y)\cup \mu(S_1^z)\cup\ldots\cup\mu(S_k^z)$.
	\end{itemize}
 \smallskip
 
\noindent	If $r$ is a tgd or egd that can be applied to $\I_1$ with  $\mu$, we say that $\I_2$ is the \emph{result of applying} $r$ to $\I_1$ with $\mu$ and we write $\I_1 \xrightarrow[]{r,\mu} \I_2$. We call $\I_1 \xrightarrow[]{r,\mu} \I_2$ a \emph{chase step}. 
\end{definition}

For both the entity-egd step and the value-egd step, the chase procedure constructs $\I_2$ by replacing some facts of $\I_1$. However, whereas for entity-egds the replacement is ``global'' (i.e., the two equivalence classes merged in the step are substituted by their union everywhere in $\I_1$), for value-egds the replacement is ``local'', in the sense that the two sets merged in the step are substituted by their union \emph{only} in facts occurring in the image $\mu(\phi(\vec{x}))$, which is a subset of $\I_1$.


\begin{example}
    \label{ex:chase-steps}
    Consider again the KB $\K=(\T,\DB)$ of Ex.~\ref{ex:TBox} and Ex.~\ref{ex:DB}. 
%
The interpretation $\I$ of Ex.~\ref{ex:interpretation} is the base instance $\I^\DB$ for $\K$. 
%
%
We depict below the application of the rules of Def.~\ref{def:chaseStep}, starting from the instance $\I=\I^\DB$.
    
\noindent   \textbf{tgd application:} $\I^\DB \xrightarrow[]{r_4,\mu_0} \I_1$, where $\mu_0$ is such that $\mu_0(\Body(r_4))= \{ d_3 \}$. The resulting instance $\I_1$ contains the facts $d_1$-$d_6$ (see.\ Ex.~\ref{ex:interpretation}), as well as the  facts:
    \begin{align*}
		\arraycolsep=3pt
        \begin{array}{ll}
            (d_7) \Emp([\Doe_3],[\textnormal{\textbf{e}}_1^\bot]), & (d_8) \CEO([\textnormal{\textbf{e}}_1^\bot], [\textnormal{\textbf{e}}_2^\bot]).
		\end{array}
	\end{align*}
    \noindent \textbf{entity-egd application:} $\I_1 \xrightarrow[]{r_1,\mu_1} \I_2$, where   the assignment $\mu_1$ is such  that 
    $\mu_1(\Body(r_1))= \{ d_1, d_2,\JaccardSim(\{\texttt{J. Doe}\},\{\texttt{John Doe}\},\{0.6\})\}$. The resulting instance $\I_2$ contains the facts $d_3$, $d_5$, $d_7$, $d_8$, as well as the facts:
    \begin{align*}
		\arraycolsep=1.0pt
        \small
		\begin{array}{ll}
			(d_1) \rightarrow (d_9)     & \CI([\Doe_1,\Doe_2], \{ \texttt{J. Doe} \}, \{ \PNumberA \} ) \\
			(d_2) \rightarrow (d_{10}) & \CI([\Doe_1,\Doe_2], \{ \texttt{John Doe} \}, \{ \PNumberB \} ) \\
            (d_4) \rightarrow (d_{11}) & \Emp([\Doe_1,\Doe_2], [\Yahoo]) \\
           ( d_6) \rightarrow (d_{12}) & \CEO([\Yahoo],[\Doe_1,\Doe_2]).\\
		\end{array}
    \normalsize
	\end{align*}
    
    \smallskip
    \noindent
    \textbf{value-egd application:} $\I_2 \xrightarrow[]{r_2,\mu_2} \I_3$, where $\mu_2$ is such that $\mu_2(\Body(r_2))=\{ d_9, d_{10} \}$. The resulting instance $\I_3$ contains the facts $d_3$, $d_{11}$, $d_5$, $d_{12}$, $d_7$, $d_8$, as well as the facts:
    \begin{align*}
		\arraycolsep=1.0pt
        \small
		\begin{array}{ll}
			(d_9) \rightarrow (d_{13})     & \CI([\Doe_1,\Doe_2], \{ \texttt{J. Doe,John Doe} \}, \{ \PNumberA \} ) \\
			(d_{10}) \rightarrow (d_{14}) & \CI([\Doe_1,\Doe_2], \{ \texttt{J. Doe,John Doe} \}, \{ \PNumberB \} ). \\
    \end{array}
    \normalsize
	\end{align*}
\end{example}

We now define the notion of a chase sequence.

\begin{definition}
	\label{def:chaseSequence}
	Let $\K=(\T,\DB)$ be a KB. A \emph{chase sequence for $\K$} is a (finite or infinite) sequence  $\sigma=\I_0,\I_1, \I_2\dots$, such that $\I_0=\I^\DB$ and $\I_i \xrightarrow[]{r_i,\mu_i} \I_{i+1}$ is a chase step,  
	for each consecutive pair ($\I_i,\I_{i+1}$) in the sequence $\sigma$. 
\end{definition}

In the Appendix, we provide an entire chase sequence for the KB $\K$ of our running example; the first three instances of this sequence have been described in Ex.~\ref{ex:chase-steps}.

We now present some basic properties of the instances in a chase sequence.

\begin{restatable}{proposition}{chaseinterpretation}
	\label{prop:chase-interpretation}
Let $\K$ be a KB.
 \begin{itemize}
 \item 
 If $\I$ is an instance for $\K$ \wrt an equivalence relation $\sim$ on $\underD_E(\I)$  and if $\I'$ is such that $\I \xrightarrow[]{r,\mu} \I'$, then there is an equivalence relation $\sim'$ 
 such that $\I'$ is an instance for $\K$ \wrt $\sim'$ on $\underD_E(\I')$. 
\item 
 Every $\I_i$ occurring in a chase sequence for a KB $\K=(\T,\DB)$ is an instance for $\K$.  
\end{itemize}
 
 \commentout{Moreover, $\sim^1 \subseteq \sim^2$ and $\underD(\I_1) \subseteq \underD(\I_2)$.}
\end{restatable}

\ifthenelse{\equal{\showproofs}{1}}{
    \begin{proof}
	We prove the thesis at the first bullet by considering the three possible cases for the rule $r$ and the corresponding chase steps of Definition~\ref{def:chaseStep}.
	
	\emph{Case of tgd chase step:} Let $r$ be a tgd of the form (\ref{eq:tgd}). Let $\S=\underD_E(\I) \cup \textsf{N}^{f_e}$, where $\textsf{N}^{f_e}=\{f_e^1,\ldots,f_e^h\}$, i.e., $\textsf{N}^{f_e}$ is the set 
    of the fresh null symbols used in Definition~\ref{def:chaseStep} for the tgd case. Let $\sim'=\sim \cup \{\tup{f,f}~|~f \in \textsf{N}^{f_e}\}$	(notice that $\textsf{N}^{f_e}$ 
		could even be empty, which implies $\sim'=\sim$). It is easy to see that $\sim'$ is an equivalence relation over $\S$, and that $\underD_E(\I)/{\sim} \subseteq \S/{\sim'}$. Then, according to Definition~\ref{def:chaseStep} and thus for each $P(\genericSet_1,\ldots,\genericSet_n) \in \I'$  and each $1 \leq i \leq n$ such that $\type(P,i)=\entity$, either $T_i \in \underD_E(\I)/{\sim}$, and thus  $\genericSet_i \in  \S/{\sim'}$, or $\genericSet_i =\{f\}$, where $f \in \textsf{N}^{f_e}$, and from the definition of $\sim'$ it obviously holds that $\{f\} \in  \S/{\sim'}$. Notice also that $\S=\underD_E(\I')$.  Moreover, for each $1 \leq i \leq n$ such that $\type(P,i)=\val$, it holds either $T_i \subseteq \underD_V(\I)$ or $T_i=\{f'\}$, where $f' \in \{f_v^1,\ldots,f_v^j\} \subseteq \alphaValueNulls$. Thus, $T_i$ is a non-empty subset of $(\alphaValues \cap \sig(\K))\cup \alphaValueNulls$. The above facts show that $\I'$ is an instance for $\K$ \wrt $\sim'$ over $\S=\underD_E(\I')$. 
		
	\emph{Case of entity-egd chase step:} 
Let $r$ be an entity-egd of the form (\ref{eq:egd}).  In this case we have $\underD_E(\I')=\underD_E(\I)$ and $\sim'=\sim \cup \sim^e$, where $\sim^e = \{\tup{f,g} \mid f,g \in \mu(y) \cup \mu(z)\}$. Obviously, $\sim^{e}$ is an equivalence relation over $\mu(y) \cup \mu(z)$. Thus $\sim'$ is the union of an equivalence relation over $\underD_E(\I)$ and an equivalence relation over $\mu(y) \cup \mu(z)$. 
As a consequence, $\sim'$ is a reflexive and symmetric  relation over $ \underD_E(\I) \cup \mu(y) \cup \mu(z)$, and since $\mu(y) \cup \mu(z) \subseteq \underD_E(\I)$, then $\sim'$ is reflexive and symmetric  over $\underD_E(\I)=\underD_E(\I')$. We have then to prove that  $\sim'$ is also transitive. 
To this aim we consider all possible cases. Let $a,b,c$ be elements in $\S$. The following situations are conceivable:
	\begin{itemize}
		\item $(a,b) \in \sim$ and $(b,c) \in \sim$, but since $\sim$ is an equivalence relation over $\underD_E(\I)$,  then  $(a,c) \in \sim$, and since $\sim \subseteq \sim'$ then $(a,c) \in \sim'$.
		\item $(a,b) \in \sim^{e}$ and $(b,c) \in \sim^{e}$, but since $\sim^{e}$ is an equivalence relation over $\mu(y) \cup \mu(z)$, then $(a,c) \in \sim^{e}$ and since $ \sim^{e}\subseteq \sim'$ then $(a,c) \in \sim'$.
		\item $(a,b) \in \sim$ and $(b,c) \in \sim^{e}$.  Since  $(b,c) \in \sim^{e}$, by definition of chase step, $b,c \in \mu(y) \cup \mu(z)$, and also $[b]_{\sim}\subseteq \mu(y) \cup \mu(z)$. Since  $(a,b) \in \sim$, it also holds that  $a \in [b]_{\sim}$, and thus  $a\in \mu(y) \cup \mu(z)$. This means that $(a,b) \in \sim^{e}$ as well, which coincides with the previous case (for which we have already shown that transitivity holds).
	\end{itemize}
	We can then prove that each fact $P(T_1,\ldots,T_n)$ of $\I'$ satisfies the conditions of Definition~\ref{def:instance} with an argument similar to that used in the last part of the proof for the tgd case. The only difference is that if some $T_i$ does not occur also in some fact of $\I$, then it must be $T_i=\mu(z) \cup \mu(y)$,  
	which is indeed an equivalence class of $\underD_E(\I')/\sim'$. 
	
	\emph{Case of value-egd chase step:} 
Let $r$ be a value-egd of the form (\ref{eq:egd}). In this case $\underD_E(\I')=\underD_E(\I)$ and $\sim'=\sim$. 
	The rest of the proof proceeds as for the previous cases. We however note that in this case each $T_i$ that does not occur in a fact of $\I$ must be of the form $\mu(S_1^y) \cup \ldots \cup \mu(S_m^y)\cup \mu(S_1^z)\cup\ldots\cup\mu(S_k^z)$, which is clearly a non-empty subset of $(\alphaValues \cap \sig(\K))\cup \alphaValueNulls$. 

\smallskip

The property at the second bullet easily follows from the property at the first bullet and the fact that in every chase sequence $\I_0 = \I^\DB$ and $\I^\DB$ is an instance for $\K$ by construction.
 
\end{proof}
}


\commentout{A chain of chase steps constitutes a chase sequence, as formally defined below.}



\commentout{

The following proposition easily follows from Prop.~\ref{prop:chase-interpretation} and the fact that in every chase sequence $\I_0=\I^\DB$ and  $\I^\DB$ is an instance for $\K$ by construction. 

\begin{proposition}
	\label{pro:chase-step-interpretation}
	Every $\I_i$ occurring in a chase sequence for a KB $\K=(\T,\DB)$ is an instance for $\K$.   
\end{proposition}
}


\begin{restatable}{proposition}{containment}
	\label{pro:containment}
	Let 
 $\K$ be a KB, let $\sigma$ be a chase sequence for $\K$, and let $\I_i$ and $\I_j$, with $i\leq j$, be two instances for $\K$ \wrt equivalence relations $\sim^i$ and $\sim^j$, respectively, such that $\I_i$ and $\I_j$  belong to $\sigma$. Then the following holds:
 \begin{itemize}
     \item  $\sim^i \subseteq \sim^j$
     and  $\underD(\I_i) \subseteq \underD(\I_j)$;
 \item If  $P(\vec{\genericSet})$ is a fact in  ${\I_i}$, then there is
	a fact  $P(\vec{\genericSet'})$ in $\I_j$  such that $P(\vec{\genericSet}) \leq P(\vec{\genericSet}')$.
 \end{itemize}
\end{restatable}

\ifthenelse{\equal{\showproofs}{1}}{
    \begin{proof}
	We recall that $\K=(\T,\DB)$, where $\T$ is the TBox of $\K$ and $\DB$ the database of $\K$. The properties can be easily proved by induction.
	\emph{Base case.} Let be $i=0$ and $j=1$, that is, consider $\I_0$ and $\I_1$, where $\I_0$ is the base instance $\I^\DB$ of $\K$ and $\I_1$ is the instance produced by the application of a chase step to $\I^\DB$. The inclusions $\sim^0 \subseteq \sim^1$ and $\underD(\I_0) \subseteq \underD(\I_1)$ can be proved by using the same construction adopted in the proof of Proposition~\ref{prop:chase-interpretation}, where $\I=\I_0$ and $\I'=\I_1$. In particular, for the case in which $r$ is a tgd, by construction, we have that $\sim_0=\sim \subseteq \sim'=\sim_1$. Moreover, $\underD_V(\I')=\underD_V(\I) \cup \{f_v^1,\ldots,f_v^j\}$ and thus $\underD(\I) \subseteq \underD(\I')$ (that is, $\underD(\I_0) \subseteq \underD(\I_1)$. The cases in which $r$ is a value-egd or an entity-egd are even more straightforward.
 
 It remains to prove that for each fact $P(\vec{\genericSet}) \in \I_0$ there exists  $P(\vec{\genericSet'}) \in \I_1$  such that $\vec{\genericSet} \leq \vec{\genericSet'}$. Let us consider the three possible chase steps. \emph{(tgd)}: this step does not modify facts in $\I_0$, thus $\vec{\genericSet}=\vec{\genericSet'}$. \emph{(entity-egd)}: every equivalence class $E$ in $\activeD(\I_0)$, i.e., the active domain of $\I_0$, remains as is or it is substituted by a class $E'$ such that $E\subseteq E'$. Moreover, $\activeD_V(\I_0)=\activeD_V(\I_1)$. These facts easily show the thesis in this case. \emph{(value-egd)}: the argument is similar as the previous case, that is, every set of values $\valueSet$ in $\activeD_V(\I_0)$ remains as is or it is substituted by a set of values $\valueSet'$ such that $\valueSet\subseteq \valueSet'$. Moreover, $\activeD_E(\I_0)=\activeD_E(\I_1)$. 

	\emph{Inductive case.} we assume that the thesis holds for a certain integer $\ell$, i.e., that properties $(i)$, $(ii)$, and $(iii)$ hold for each $i,j \leq \ell$ such that $i \leq j$, and prove that they hold also for each $i,j \leq \ell+1$ such that $i \leq j$. The only cases that need to be shown are those in which $j=\ell+1$, because the other cases already hold by the inductive hypothesis.
	Of course, if $i =\ell+1$ the condition is trivially satisfied. For $i=\ell$ the three properties can be proved exactly as done for the base case. Then, all cases in which $i<\ell$ follow by the previous case and the inductive hypothesis.
\end{proof}
}

 Prop.~\ref{pro:containment} implies that for every $i\leq j$ and every equivalence class $E$ in $\underD_E(\I_i)/{\sim^i}$, there exists $E'$ in  $\underD_E(\I_j)/\hspace*{-3pt}\sim^j$ such that $E \subseteq E'$;  in particular, we have that $[e]_{\sim^i} \subseteq [e]_{\sim^j}$, for each entity $e \in \alphaEntities \cap \sig(\K)$.

When a chase sequence for a KB is infinite, there might be rules applicable to some instance $\I_i$ in the sequence that are never applied, even though they remain applicable in subsequent chase steps. It is however always possible to establish a suitable order in the application of the rules so that the above situation never occurs. Chase sequences enjoying this property are called \emph{fair}. In what follows, we assume that all chase sequences considered are fair. 

We now investigate the possible outcomes of the chase for a KB $\K$. 
\commentout{We first deal with the case in which the chase procedure terminates and then consider the case in which it is constituted by infinitely many steps.}
We first give the definition of the result of the chase for a finite chase sequence.


\begin{definition}
	\label{def:chase-finite-case}
	Let $\K=(\T,\DB)$ be a KB. 
 \begin{itemize}
\item  A \emph{finite chase of $\K$} is a finite chase sequence $\sigma=\I_0,\I_1,\ldots,\I_m$, 
 for which  there is no rule $r$ in $\T$ and no assignment $\mu$ such that $r$ can be applied to $\I_m$ with $\mu$. 
 
 \item We say that $\I_m$ is the \emph{result of the finite chase  $\sigma$ for $\K$}, and denote it by $\chase(\K,\sigma)$.
 \end{itemize}
\end{definition}

Prop.~\ref{prop:chase-interpretation} implies  that $\chase(\K,\sigma)$ is an instance for $\K$. Actually, it turns out to be a solution for $\K$.

\begin{restatable}{lemma}{chasefinitecasemodel}
	\label{lem:chase-finite-case-model}
	Let $\K$ be a KB and let $\sigma=\I_0,\I_1,\ldots,\I_m$ be a finite chase for $\K$. Then $\chase(\K,\sigma)$ is a solution for $\K$.
\end{restatable}

\ifthenelse{\equal{\showproofs}{1}}{
    \begin{proof}
	We need to prove that $\chase(\K,\sigma)$ satisfies all facts in $\DB$, all tgds, entity-egds and value-egds in $\T$. We show each point separately:
	\begin{itemize}
		\item[(i)] The property that $\chase(\K,\sigma)$ satisfies all facts in $\DB$ follows from $\I_0=\I^\DB$ and from Proposition~\ref{pro:containment}. Indeed, $\I^\DB$ satisfies all facts in $\DB$, by definition of base instance, which means that for each ground atom $P(c_1,\ldots,c_n)$ in $\DB$ there is a fact $P(\genericSet_1,\ldots,\genericSet_n)$ in $\I_0$ such that $c_i \in T_i$ for $1 \leq i\leq n$. By Proposition~\ref{pro:containment} it follows that in $\I_m$ there is a tuple $\tup{T_1',\ldots,T_n'}$ such that $T_1 \subseteq T_1',\ldots,T_n \subseteq T_n'$, and thus $P(c_1,\ldots,c_n)$ is satisfied also in $\I_m=\chase(\K,\sigma)$.
		\item[(ii)] To prove that $\chase(\K,\sigma)$ satisfies all tgds in $\T$, let us assume by contradiction that there is a tgd $d = \forall \vec{x}(\phi(\vec{x}) \rightarrow \exists \vec{y} \psi(\vec{x}, \vec{y}))$ that is not satisfied by $\chase(\K,\sigma)$. This means that there is an assignment $\mu$ from $\phi(\vec{x})$ to $\chase(\K,\sigma)$ such that there is no $\psi(\vec{x},\vec{y})$-compatible tgd-extension $\mu'$ of $\mu$ to $\chase(\K,\sigma)$. But thus $d$ would be applicable to $\chase(\K,\sigma)$ with assignment $\mu$, which contradicts the hypothesis that $\sigma$ is a finite chase.
		\item[(iii)] The fact that $\chase(\K,\sigma)$ satisfies all the entity-egds and value-egds can be proved as done in (ii) for satisfaction of tgds in $\T$.
	\end{itemize}
\end{proof}
}

The following lemma is used to prove that if $\sigma$ is a finite chase, then $\chase(\K,\sigma)$ is not just a solution for $\K$, but also a universal solution for $\K$. A similar result, known as the Triangle Lemma, was used in~\cite{FKMP05} in the context of data exchange.


\begin{restatable}{lemma}{trianglefinitecase}
	\label{lem:triangle-finite-case}
	Let $\I_{1} \xrightarrow[]{r,\mu} \I_{2}$ be a chase step.
	Let $\I$ be an instance for $\K$ such that  $\I$ satisfies $r$ and  there exists a homomorphism $h_1: \I_1\rightarrow \I$. Then there is a homomorphism $h_2: \I_2\rightarrow \I$ such that $h_2$ extends $h_1$.
\end{restatable}

\ifthenelse{\equal{\showproofs}{1}}{
    \begin{proof}
	
	The proof follows the line of the proof of Lemma 3.4 in \cite{FKMP05}. However, since our structures are based on equivalence classes and sets of values, the various steps of the proof are more involved. In particular, differently from~\cite{FKMP05}, in our framework the composition of an assignment with an homomorphism does not directly yield a new assignment. Nonetheless, an assignment can be constructed on the basis of such composition.
	
	We prove the thesis by considering the three possible cases for the rule $d$ and the corresponding chase steps of Definition~\ref{def:chaseStep}.
	
	\emph{Case of tgd chase step:} Let $r$ be a tgd of the form (\ref{eq:tgd}). For ease of exposition, we  assume $\phi(\vec{x}) = P_1(x_1,x_2) \land P_2(x_1,x_2)$ and $\psi(\vec{x},\vec{y})=P_3(x_1,x_2,y_1,y_2) \land P_4(x_2,y_1,y_2)$, where $x_1$ and $y_2$ are entity variables and $x_2$ and $y_1$ are value-variables. 
 The case of generic tgds can be managed in a similar way.
	If we apply the operator $\tau$ to both the body and the head of the tgd, we obtain the following formulas:\\
	$\tau(\phi(\vec{x}))=P_1(x_1,S_1^{x_2})\land P_2(x_1,S_2^{x_2})$\\
	$\tau(\psi(\vec{x},\vec{y}))=P_3(x_1,R_1^{x_2},R_1^{y_1},y_2)\land P_4(R_2^{x_2},R_2^{y_1},y_2)$\\
	We first construct an assignment $\hat{\mu}$ from $\phi(\vec{x})$ to $\I$.   By the definition of chase step, $\mu$ is an assignment from $\phi(\vec{x})$ to $\I_1$. Then, there exist an equivalence class $E_1 \in \activeD_E(\I_1)$ and two sets of values $V_1,V_2 \in \activeD_V(\I_1)$ such that $V_1 \cap V_2 \neq \emptyset$, $\mu(x_1)=E_1$, $\mu(S_1^{x_2})=V_1$, $\mu(S_2^{x_2})=V_2$, $P_1(E_1,V_1)$ and $P_2(E_1,V_2)$ are facts in $\I_1$.
	Since $h_1$ is homomorphism from $\I_1$ to $\I$, there exist facts $P_1(E_1',V_1')$, $P_2(E_2',V_2')$ in $\I$, 
	such that 
    $h_1(P_1(E_1,V_1)) \leq P_1(E_1',V_1')$ and $h_1(P_2(E_1,V_2)) \leq P_2(E_1',V_2')$, that is,
    $h_1(E_1)\subseteq E_1'$, $h_1(E_1)\subseteq E_2'$, $h_1(V_1)\subseteq V_1'$, $h_1(V_2)\subseteq V_2'$. 
	Since $V_1 \cap V_2 \neq \emptyset$, then also $V_1' \cap V2' \neq \emptyset$. Moreover, since $E_1'$ and $E_2'$ are equivalence classes containing the same non-empty subset, they must be the same equivalence class, i.e., $E_1'=E_2'$. We then define the mapping $\hat{\mu}$  as follows: $\hat{\mu}(x_1)=E_1'$, $\hat{\mu}(S_1^{x_2})=V_1'$, and $\hat{\mu}(S_2^{x_2})=V_2'$. It is easy to see that $\hat{\mu}$ satisfies the conditions in Definition~\ref{def:assignment-revisited}, and thus it  is an assignment from $\phi(\vec{x})$ to $\I$.
	
	Since $\I$ satisfies $r$, there is an assignment $\mu'$ that is a $\psi(\vec{x},\vec{y})$-compatible extension of $\hat{\mu}$ to $\I$. By definition we have that: 
	\begin{itemize}
		\item $\mu'(x_1)=\hat{\mu}(x_1)=E_1'$
		\item $\mu'(R_1^{x_2})=V_3'$, $\mu'(R_2^{x_2})=V_4'$, where $V_3',V_4' \in \activeD_V(\I)$ are such that $V_1' \cap V_2' \subseteq V_3' \cap V_4'$,
		\item $\mu'(R_1^{y_1})=V_5'$, $\mu'(R_2^{y_1})=V_6'$, where $V_5',V_6' \in \activeD_V(\I)$ are such that $V_5' \cap V_6' \neq \emptyset$,
		\item $\mu'(y_2)=E_2'$, where $E_2' \in \activeD_E(\I)$,
		\item $\I$ contains the facts $P_3(E_1',V_3',V_5',E_2')$ and $P_4(V_4',V_6',E_2')$.
	\end{itemize}

	Then, by definition of tgd chase step it follows that there is also an assignment $\mu_2$ from $\psi(\vec{x},\vec{y})$ to $\I_2$ defined as follows
	\begin{itemize}
		\item $\mu_2(x_1)=\mu(x_1)=E_1$
		\item $\mu_2(R_1^{x_2})=\mu_2(R_2^{x_2})=V_1 \cap V_2$
		\item $\mu_2(R_1^{y_1})=\mu_2(R_2^{y_1})=\{f_v^{y_1}\}$
		\item $\mu_2(y_2)=[f_e^{y_2}]$
	\end{itemize}
	where $f_v^{y_1}$ and $f_e^{y_2}$ are the fresh value-null and entity-null, respectively, introduced by the chase step. We remark that $\mu_2$ is a $\psi(\vec{x},\vec{y})$-compatible extension of $\mu$ to $\I_2$, and that $\underD(\I_2)=\underD(\I_1) \cup \{f_v^{y_1},f_e^{y_2}\}$. Moreover, $\I_2= \I_1 \cup \{P_3(E_1,V_1 \cap V_2, \{f_v^{y_1}\},[f_e^{y_2}]),P_4(V_1 \cap V_2, \{f_v^{y_1}\},[f_e^{y_2}])\}$.

	We now construct a mapping $h_2$ 
	and show that it is an homomorphism from $\I_2$ to $\I$.
	We define $h_2$ as follows:
	\begin{itemize}
		\item $h_2(o)=h_1(o)$ for each $o \in \underD(\I_1)$;
		\item $h_2(f_v^{y_1})=v$, where $v$ is any element in $V_5' \cap V_6'$;
		\item $h_2(f_e^{y_2})=e$, where $e$ is any element in $E_2'$.
	\end{itemize}
	Note that, as stated before, $V_5' \cap V_6' \neq \emptyset$. Moreover, $E_2' \neq \emptyset$, since $E_2'$ is an equivalence class in $\activeD(\I)$. Thus $v$ and $e$ above always exist. 
	Note also that $h_2$ extends $h_1$.
	
	By construction, $h_2$ is a mapping from the elements of $\underD(\I_2)$ to the elements of $\underD(\I)$. Moreover, Conditions~1-4 of Definition~\ref{def:homomorphism} (homomorphism) are trivially satisfied.
	To prove that $h_2$ is homomorphism from $\I_2$ to $\I$ it remains to prove that also Condition~5 of Definition~\ref{def:homomorphism} holds. 
	Such condition is obviously satisfied by all facts in $\I_1$ (which is contained in $\I_2$), since $h_2$ over such facts behaves as $h_1$, and $h_1$ is an homomorphism to $\I$. 
	Consider now the fact $P_3(E_1,V_1 \cap V_2, \{f_v^{y_1}\},[f_e^{y_2}])$, which is in $\I_2\setminus \I_1$. We have to show that there exists a fact in $\I$ of the form $P_3(\hat{E_1},\hat{V_1},\hat{V_2},\hat{E_2})$, where $\hat{E_1},\hat{E_2} \in \activeD_E(\I)$ and $\hat{V_1},\hat{V_2} \in \activeD_V(\I)$, 
	such that 
    $h_2(P_3(E_1,V_1 \cap V_2, \{f_v^{y_1}\},[f_e^{y_2}])) \leq P_3(\hat{E_1},\hat{V_1},\hat{V_2},\hat{E_2})$, i.e.,
    $h_2(E_1) \subseteq\hat{E_1}$, $h_2(V_1 \cap V_2)\subseteq \hat{V_1}$, $h_2( \{f_v^{y_1}\})\subseteq \hat{V_2}$,  $h_2( [f_e^{y_2}])\subseteq \hat{E_2}$.
	To this aim we set $\hat{E_1}=E_1'$, $\hat{V_1}=V_3'$, $\hat{V_2}=V_5'$, $\hat{E_2}=E_2'$, and show that the above mentioned containments do indeed hold.
	\begin{itemize}
		\item $h_2(E_1) \subseteq E_1'$ follows from $h_2(E_1)=h_1(E_1)$ and $h_1(E_1) \subseteq E_1'$;
		\item $h_2(V_1 \cap V_2) \subseteq V_3'$ is a consequence of $h_2(V_1 \cap V_2)=h_1(V_1 \cap V_2) \subseteq h_1(V_1) \cap h_1(V_2) \subseteq V_1' \cap V_2' \subseteq V_3' \cap V_4' \subseteq V_3'$. 
		\item $h_2(\{f_v^{y_1}\}) \subseteq V_5'$ holds because $h_2(f_v^{y_1})=v$ and $v \in V_5' \cap V_6'$, and thus $h_2(\{f_v^{y_1}\})=\{v\} \subseteq V_5'$.
		\item $h_2([f_e^{y_2}]) \subseteq E_2'$ holds because $h_2(f_v^{y_1})=e$ and $e \in E_2'$, and thus $h_2([f_e^{y_2}])=[e] \subseteq E_2'$.
	\end{itemize}
	We can proceed in analogously to show that Condition~5 of Definition~\ref{def:homomorphism} is satisfied also for the other fact, $P_4(V_1 \cap V_2, \{f_v^{y_1}\},[f_e^{y_2}])$, which is contained in $\I_2 \setminus \I_1$.
	
	\emph{Case of entity-egd chase step:} 
	Let $r$ be an entity-egd of the form (\ref{eq:egd}). 
	By the definition of the chase step, $\mu$ is an assignment from $\phi(\vec{x})$ to $\I_1$, and in particular there exist equivalence classes $E_1,E_2 \in \activeD_E(\I_1)$ such that $\mu(y)=E_1$ and  $\mu(z)=E_2$.
	It is then possible to construct an assignment $\hat{\mu}$ from $\phi(\vec{x})$ to $\I$ such that $\hat{\mu}(y)=E_1'$, $\hat{\mu}(z)=E_2'$, $E_1',E_2' \in \activeD_E(\I_1)$, $h_1(E_1) \subseteq E_1'$ and $h_1(E_2) \subseteq E_2'$.
	The construction of $\hat{\mu}$ is similar to the analogous construction described for the case of tgd chase step. For the sake of completeness we describe this procedure also for the entity-egd chase step (note that in this case we need to also deal with entities and values possibly occurring in the rule body $\phi(\vec{x})$). For ease of exposition, and without loss of generality, we assume that $\phi(\vec{x})=P_1(x,y,e,v) \land P_2(x,y,z)$, where $x$ is a value-variable, $y,z$ are entity-variables, $e \in \alphaEntities$, and $v \in \alphaValues$. By applying the operator $\tau$ to $\phi(\vec{x})$ we obtain
	$\tau(\phi(\vec{x}))=P_1(S_1^{x},y,e,v)\land P_2(S_2^{x},y,z)$. Since $\mu$ is an assignment from $\phi(\vec{x})$ to $\I_1$, there exist $E,E_1,E_2 \in \activeD_E(\I_1)$  and $V,V_1,V_2 \in \activeD_V(\I_1)$ such that $e \in E$, $V_1 \cap V_2 \neq \emptyset$, $v \in V$, $P_1(V_1,E_1,E,V)$ and $P_2(V_2,E_1,E_2)$ are two facts in $\I_1$, and $\mu(S_1^x)=V_1$, $\mu(y)=E_1$, $\mu(v)=V$, $\mu(S_2^x)=V_2$, and $\mu(z)=E_2$.
	Since $h_1$ is homomorphism from $\I_1$ to $\I$, there exist 
    $P_1(V_1',E_1',E',V')$ and $P_2(V_2',E_1',E_2')$ in $\I$ such that $h_1(P_1(V_1,E_1,E,V))\leq P_1(V_1',E_1',E',V')$ and $h_1(P_2(V_2,E_1,E_2)) \leq P_2(V_2',E_1',E_2')$, i.e.,
    $h_1(E_1) \subseteq E_1'$, $h_1(V_1) \subseteq V_1'$, $h_1(E) \subseteq E'$, $h_1(V) \subseteq V'$, $h_1(V_2) \subseteq V_2'$, $h_1(E_2) \subseteq E_2'$ and . We then define $\hat{\mu}$ as follows: $\hat{\mu}(S_1^x)=V_1'$, $\hat{\mu}(y)=E_1'$, $\hat{\mu}(v)=V$, $\hat{\mu}(S_2^x)=V_2'$, $\hat{\mu}(z)=E_2'$.
	It is easy to see that $\hat{\mu}$ is an assignment from $\phi(\vec{x})$ to $\I$. In particular, since $e \in E$ and $h(e)=e$ it obviously holds $e \in h_1(E) \subseteq E'$, which implies  $e \in E'$. Similarly, since $v \in \mu(v)=V$ and $v=h_1(v)$ we have that $v \in h_1(V) \subseteq V'$, which implies  $v \in V'$.

	We then take $h_2$ to be $h_1$ and show that $h_1$ is still a homomorphism when considered from $\I_2$ to $\I$. 
	We soon note that $\underD(\I_1)=\underD(\I_2)$, and thus $h_1$ is a mapping from $\underD(\I_2)$ as well. Moreover, since $h_1$ is an homomorphism, then Conditions~1-4 of Definition~\ref{def:homomorphism} are obviously satisfied. It remains to prove that Condition~5 of Definition~\ref{def:homomorphism} is also satisfied, i.e., we have to prove that for each fact  $P(\genericSet_1,\ldots,\genericSet_n)$ in $\I_2$ there exists a fact $P(\genericSettwo_1,\ldots,\genericSettwo_n)$ in $\I$ such that  $h_1(\genericSet_i) \subseteq \genericSettwo_i$, for $1 \leq i \leq n$.
	We recall that $\I_2$ is obtained by replacing in $\I_1$ every occurrence of $E_1$ and every occurrence of $E_2$ with $E_1 \cup E_2$. Therefore, two cases are possible. $(a)$ $\genericSet_i \neq E_1 \cup E_2$  for $1 \leq i \leq n$. In this case, the thesis holds since $P(\genericSet_1,\ldots,\genericSet_n)$ is also in $\I_1$ and $h_1$ is an homomorphism from $\I_1$ to $\I$. $(b)$ there is $j$ such that $1 \leq j \leq n$ and $\genericSet_j=E_1 \cup E_2$. Without loss of generality we can also assume that $\genericSet_i \neq E_1 \cup E_2$ for $1 \leq i \leq n$ and $i \neq j$. Since two equivalence classes are either disjoint or are the same class, we have that each $T_i$ such that $i\neq j$ is either a set of values or an equivalence class disjoint from $E_1 \cup E_2$.
	Then, there is a fact $P(\genericSet_1,\ldots,\genericSet_{j-1},\genericSet'_j,\genericSet_{j+1}\ldots,\genericSet_n)$ in $\I_1$ where it is either $\genericSet'_j=E_1$ or $\genericSet'_j=E_2$, and, for every $i\neq j$, $T_i \neq E_1$ and $T_i \neq E_2$ . Let us consider  $\genericSet'_j=E_1$ (the other choice can be treated analogously). Then, there must exist a fact $P(\genericSettwo_1,\ldots,\genericSettwo_{j-1},\genericSettwo_j,\genericSettwo_{j+1}\ldots,\genericSettwo_n)$ in $\I$ such that $h_1(T_i) \subseteq U_i$ for $1 \leq i \leq n$. In particular, this means that $h_1(E_1) \subseteq \genericSettwo_j$, with $\genericSettwo_j$ an equivalence class in $\activeD_E(\I)$. 
	Since $h_1(E_1) \subseteq E_1'$, it must be $\genericSettwo_j =E_1'$. Indeed, for the two equivalence classes $\genericSettwo_j$  and $E_1'$ it holds either $\genericSettwo_j =E_1'$ or $\genericSettwo_j  \cap E_1'=\emptyset$, but the second option is impossible, since  $h_1(E_1)$ is a common non-empty subset of $\genericSettwo_j$  and $E_1'$. Then, since $\I$ satisfies the entity-egd $d$, which in $d$ is triggered by $\hat{\mu}$, we have that $E_1'=E_2'$ and thus from $h_1(E_1)\subseteq E_1'$ and $h_1(E_2)\subseteq E_2'$ it follows  $h_1(E_1)\cup h_1(E_2)\subseteq E_1'$ and thus $h_1(E_1\cup E_2)\subseteq E_1'$, that is $h_1(\genericSet_j)\subseteq E_1'=U_j$. This shows that $P(\genericSettwo_1,\ldots,\genericSettwo_{j-1},\genericSettwo_j,\genericSettwo_{j+1}\ldots,\genericSettwo_n)$ is the fact in $\I$ we were looking for, thus proving Condition~5 of Definition~\ref{def:homomorphism}.
	
	\emph{Case of value-egd chase step:} The proof is similar to the one used for the entity-egd chase step. 
	Let $r$ be an entity-egd of the form (\ref{eq:egd}), and, Without loss of generality, assume that $\phi(\vec{x})=P_1(y,z)\land P_2(z)$, where $y,z$ are value-variables (the possible presence of entity-variables, entities and values can be dealt with as shown in the previous cases of this proof). Then, $\tau(\phi(\vec{x}))=P_1(S_1^y,S_1^z)\land P_2(S_2^z)$. Since $\mu$ is assignment from $\phi(\vec{x})$ to $\I_1$, then there exist $V_1,V_2,V_3 \in \activeD_V(\I_1) $ such that $\mu(S_1^y)=V_1$, $\mu(S_1^z)=V_2$, $\mu(S_2^z)=V_3$, $V_2 \cap V_3 \neq \emptyset$, $P_1(V_1,V_2)$ and $P_2(V_3)$ are facts in $\I_1$.
	Since $h_1$ is homomorphism from $\I_1$ to $\I$, then there are facts in $P_1(V_1',V_2')$ and $P_2(V_3')$ in $\I$ such that 
    $h_1(P_1(V_1,V_2)) \leq P_1(V_1',V_2')$ and $h_1(P_2(V_3)) \leq P_2(V_3')$, i.e.,
    $h_1(V_1) \subseteq V_1'$, $h_1(V_2) \subseteq V_2'$, $h_1(V_3) \subseteq V_3'$. Since $V_2 \cap V_3 \neq \emptyset$, then $V_2' \cap V_3' \neq \emptyset$. We define a mapping $\hat{\mu}$ such that $\hat{\mu}(S_1^y)=V_1'$, $\hat{\mu}(S_1^z)=V_2'$, $\hat{\mu}(S_2^z)=V_3'$. It is easy to see that $\hat{\mu}$ is an assignment from $\phi(\vec{x})$ to $\I$.
	
	We then take $h_2$ to be $h_1$ and show that $h_1$ is still a homomorphism when considered from $\I_2$ to $\I$. 
	We soon note that $\underD(\I_1)=\underD(\I_2)$, and thus $h_1$ is a mapping from $\underD(\I_2)$ as well. Moreover, since $h_1$ is an homomorphism, then condition 1-4 of Definition~\ref{def:homomorphism} are obviously satisfied.
	Condition~5 of Definition~\ref{def:homomorphism} is also already satisfied for all facts in $\I_2$ that are also in $\I_1$. It remains to prove that such condition holds also for facts in  $\I_2 \setminus \I_1$. The only two such facts are $P_1(V_1 \cup V_2 \cup V_3, V_1 \cup V_2 \cup V_3)$ and $P_2(V_1 \cup V_2 \cup V_3)$ (see the value-egd chase step in Definition~\ref{def:chaseStep}). We have to show that there are facts $P_1(\hat{V_1},\hat{V_2})$ and $P_2(\hat{V_3})$ in $\I$ such that 
    $h_1(P_1(V_1 \cup V_2 \cup V_3, V_1 \cup V_2 \cup V_3)) \leq P_1(\hat{V_1},\hat{V_2})$ and $h_1(P_2(V_1 \cup V_2 \cup V_3)) \leq P_2(\hat{V_3})$, i.e.,
    $h_1(V_1 \cup V_2 \cup V_3) \subseteq \hat{V_1}$, $h_1(V_1 \cup V_2 \cup V_3) \subseteq \hat{V_2}$, and $h_1(V_1 \cup V_2 \cup V_3) \subseteq \hat{V_3}$.
	It is not difficult to see that such facts are $P_1(V_1',V_2')$ and $P_2(V_3')$. Indeed, we know that $h_1(V_1) \subseteq V_1'$ , $h_1(V_2) \subseteq V_2'$ and $h_1(V_3) \subseteq V_3'$. Then, since $\I$ satisfies $d$, which is triggered in $\I$ by $\hat{\mu}$,  it must be $V_1'=V_2'=V_3'$. Thus, we have that $h_1(V_1) \cup h_1(V_2) \cup h_1(V_3) \subseteq V_1'=V_2'=V_3'$, and thus $h_1(V_1 \cup V_2 \cup V_3) \subseteq V_1'=V_2'=V_3'$, which shows the thesis.
\end{proof}
}

The following theorem is the main result of this section.

\begin{restatable}{theorem}{chaseunivmodel}
	\label{thm:chase-univ-model}
	If $\K$ is a KB and  $\sigma=\I_0,\I_1,\ldots,\I_m$ is a finite chase for $\K$, then $\chase(\K,\sigma)$ is a universal solution for $\K$. 
\end{restatable}

\ifthenelse{\equal{\showproofs}{1}}{
\begin{proof}
	From Lemma \ref{lem:chase-finite-case-model} it follows that $\chase(\K,\sigma)$ is a model for $\K$. It remains to prove that it is universal, i.e., that there exists an homomorphism from $\chase(\K,\sigma)=\I_m$ to every solution for $\K$. The proof is based on Lemma~\ref{lem:triangle-finite-case} and on the observation that the identity mapping $\id$ is a homomorphism from $\I_0$ to every $\I \in \Mod(\K)$. The thesis can be thus proved by applying Lemma~\ref{lem:triangle-finite-case} to each $\I_i,\I_{i+1}$ in $\sigma$, starting from $\I_0,\I_1$, with $h_1=\id$ as homomorphism from $\I_0$ to a solution $\I$ for $\K$. 
\end{proof}
}

Theorem~\ref{thm:chase-univ-model} implies  that if both $\sigma$ and $\sigma'$ are finite chases for $\K$, then $\chase(\K,\sigma)$ and
$\chase(\K,\sigma')$ are homomorphically equivalent. Thus, the result of a finite chase for $\K$ is unique up to homomorphic equivalence.

It is easy to verify that if $\K=(\T,\DB)$ is a KB in which every tgd in $\T$ is full (i.e., there are no existential quantifiers in the heads of the tgds in $\T$), then every chase sequence for $\K$
is finite. In particular, this holds true if $\T$ consists of egds only, which covers all  entity resolution settings.


For infinite chase sequences, the definition of the result of the chase requires some care. The typical approach adopted when only  tgds are present~\cite{KrMR19,GoOn19}, or for  settings involving \emph{separable} tgds and egds~\cite{JoKl84,CDLLR07,CaGK13}, is to define the result of an infinite chase sequence as 
the union of all facts 
generated in the various chase steps. This approach does not work for arbitrary tgds and egds since the resulting 
instance needs not satisfy all the rules and thus needs not be  a solution in our terminology. Moreover,  in our framework the union of all instances in an infinite chase sequence is not even an instance for the KB at hand, because the sets of entities in the result does not correspond to equivalence classes with respect to an equivalence relation. 

An alternative approach, which might be better suited for a setting with arbitrary tgds and egds, is to define the result of an infinite chase sequence as the instance 
containing all \emph{persistent} facts, i.e.,
all facts that are introduced in some  step in the chase sequence and are never modified in subsequent chase steps. This notion was introduced in~\cite{BeVa84}. In our setting, given a KB $\K$ and an infinite chase sequence $\sigma= \I_0,\I_1,\ldots$, this  means that we define the result $\chase(\K,\sigma)$ of the infinite chase $\sigma$ of $\K$ as follows:
\begin{multline} 
\chase(\K,\sigma)=\{f \mid \textrm{there is some } i\geq 0 \textrm{ such that } \\ f \in \I_j  \textrm{ for each } j \geq i\}. \label{bv}
\end{multline}

The following example shows that this definition does not work in our framework, because the above set might be empty (even if the database $\DB$ is non-empty). 

\begin{example}
	\label{exa:phokion}
	Let $\K=\tup{\T, \DB}$, where $\DB=\{P(1,2)\}$ and $\T$ consists of the  two rules
	\small $$(r_1) ~ P(x,y)\rightarrow P(y,z);~~ r_2)~ P(x,y) \land P(y,z) \ra y=z,$$ \normalsize
			\commentout{
			\[
			\begin{array}{ll}
				(r_1) & P(x,y) \ra P(y,z)\\
				(r_2) &	P(x,y) \land P(y,z) \ra y=z
			\end{array}	
			\]}
			with  $\type(P)=\tup{\val,\val}$. We construct an infinite chase sequence  starting with $\I_0=\{P(\{1\},\{2\})\}$ 	and by repeatedly applying  the above rules with suitable assignments in the following order: 
			$r_1,r_1,r_2,r_1,r_2,r_1,r_2,r_1\ldots$.
			
			We obtain the  infinite chase sequence $\sigma=\I_0,\I_1,\ldots$ 
			\[
                \small
			\begin{array}{lcl}
				\I_0&=&\{P(\{1\},\{2\})\}\\
				\I_1&=&\{P(\{1\},\{2\}),P(\{2\},\{v_\bot^1\})\}\\
				\I_2&=&\{P(\{1\},\{2\}),P(\{2\},\{v_\bot^1\}), P(\{v_\bot^1\},\{v_\bot^2\})\}\\
				\I_3&=&\{P(\{1\},\{2,v_\bot^1\}),P(\{2,v_\bot^1\},\{2,v_\bot^1\}), \\
                    && P(\{v_\bot^1\},\{v_\bot^2\})\}\\
	\ldots \\			
				\I_7&=&\{P(\{1\},\{2,v_\bot^1,v_\bot^2\}),P(\{2,v_\bot^1\},\{2,v_\bot^1,v_\bot^2\}), \\ 
                    && P(\{2,v_\bot^1,v_\bot^2\},\{2,v_\bot^1,v_\bot^2\}),P(\{v_\bot^2\},\{v_\bot^3\}),\\ 
                    && P(\{v_\bot^3\},\{v_\bot^4\})\}\\
				\ldots
                \normalsize
			\end{array}
			\]
\commentout{			
			We point out that in the step $\I_2 \xrightarrow[]{r_2,\mu_2} \I_3$, the image of $\mu_2$ is given by the facts $\{P(\{1\},\{2\}),P(\{2\},\{v_\bot^1\})$; in the step $\I_4 \xrightarrow[]{r_2,\mu_4} \I_5$, the image of $\mu_4$ is given by the facts $\{P(\{1\},\{2,v_\bot^1\}),P(\{v_\bot^1\},\{v_\bot^2\})$; in the step $\I_6 \xrightarrow[]{r_2,\mu_6} \I_7$, the image of $\mu_6$ is given by the facts $P(\{2,v_\bot^1\},\{2,v_\bot^1\}),P(\{2,v_\bot^1,v_\bot^2\},\{2,v_\bot^1,v_\bot^2\})$.\nb{To better describe the example, we should indicate a policy for the selection of the assignment triggering $r_2$, because everytime there are various assignments (I thought it was the assignment containing the fact that has been modified or introduced less recently, but typically there is more than one assignment with this characteristic. }}
			
			It is not difficult to see that 
			the set $\chase(\K,\sigma)$  defined in (\ref{bv}) above is empty, i.e., \emph{no} persistent facts occur in $\sigma$.
			\qed
		\end{example}
The infinite chase sequence $\sigma$  in Ex.~\ref{exa:phokion} is fair. At the same time, there are finite fair sequences for the KB $\K$ in this example. Indeed, if we  apply rule $r_1$ and then  rule $r_2$,
 we get a finite chase sequence $\sigma' = \I'_0,\I'_1,\I'_2$, where
\[	
\small
\begin{array}{lcl}
				\I'_0&=&\{P(\{1\},\{2\})\}\\
				\I'_1&=&\{P(\{1\},\{2\}),P(\{2\},\{v_\bot^1\})\}\\
				\I'_2 & = & 
			\{P(\{1\},\{2,v_\bot\}),P(\{2,v_\bot\},\{2,v_\bot\})\}.
				\end{array}
    \normalsize
				\]
Consequently, $\chase(\K,\sigma')=\I'_2$.
Thus, further investigation is needed when both finite and infinite chase sequences exist for a KB. We leave this as a topic for future work.

\section{Related Work}
\label{sec:related}

From the extensive literature on entity resolution\ifthenelse{\equal{\showproofs}{0}}{ and due to space limitations}, we comment briefly on only a small subset of earlier work that is related to ours.

Swoosh~\cite{BGMS*09} is  a generic approach to entity resolution in which the functions used to compare and merge records are ``black boxes". 
In our framework,  the match function is determined by entity-egds and value-egds, whereas the merge function is implemented via the union operation - an important special case of merge functions in the Swoosh approach. In this sense, our framework is less general than Swoosh. In another sense, our framework is more general as it supports tgds, differentiates between entities and values, and incorporates query answering.

As in the work  on matching dependencies (MDs)~\cite{Fan08,BeKL13}, we consider variables in the head of egds to be matched and merged rather than to be made equal. 
In ~\cite{BeKL13},  generic functions obeying some natural properties are used to  compute the result of a match, 
while we use the union of the sets of values or sets of entities involved in the match. In the 
MDs framework,
a version of the chase procedure is used to resolve dependency violations.   
This chase 
acts locally on a pair of tuples at each step, whereas our chase procedure  matches entities in a global fashion and values in a local fashion. Furthermore, we support tgds, which are not in the framework of MDs.

There is a body of work on declarative entity resolution and its variants that is related to our framework; relevant references include the Dedupalog framework~\cite{ArRS09}, the declarative framework for entity linking in \cite{DBLP:journals/tods/BurdickFKPT16,DBLP:journals/jcss/BurdickFKPT19}, and the more recent LACE framework~\cite{BiCG22}. In both Dedupalog and LACE, there is a distinction between hard and soft entity resolution rules.  LACE supports global merges and creates equivalence classes  as we do,  but it does not support  local merges; also, LACE combines entity resolution rules with denial constraints, while our framework combines entity resolution rules with tgds.
The declarative framework for entity linking in \cite{DBLP:journals/tods/BurdickFKPT16,DBLP:journals/jcss/BurdickFKPT19} uses key dependencies and disjunctive constraints with ``weighted" semantics that measure the strength of the links.
The key dependencies are interpreted as hard rules that the solutions must satisfy  in the standard sense. Consequently, two conflicting links cannot co-exist in the same solution, hence that approach uses repairs to define the notion of the certain links. In contrast, in our framework we carry along via the chase all the alternatives as either sets of values or equivalence classes of entities. Another feature of our framework is the focus on the certain answers of queries.
Finally, the aforementioned declarative framework for entity linking uses more general link relations without rules for an equivalence relation of entities; thus, it is less focused on entity resolution. 
\section{Conclusions and Future Work}
\label{sec:conclusions}


The main contribution of this paper is the development of a new declarative framework that combines entity resolution and query answering in knowledge bases. This is largely a conceptual contribution because the development of the new declarative framework entailed rethinking from first principles  the definitions of such central notions as assignment, homomorphism, satisfaction of tgds and egds in a model, and universal solution. At the technical level, we designed a chase procedure that never fails, and showed that, when it terminates, it produces a universal solution that, in turn, can be used in query answering.

As regards future  directions, perhaps the most pressing issue is to identify a ``good" notion of the result of the chase when the chase procedure does not terminate. This may lead to extending the framework presented here to settings where all universal solutions are infinite. 
While infinite universal solutions cannot be materialized,
such solutions have been used to obtain rewritability results~\cite{CDLLR07},  occasionally combined with partial materialization of the result of an infinite chase~\cite{LuTW09}. In parallel, it is important to identify structural conditions on the tgds and the egds of the TBox that guarantee termination of the chase procedure in polynomial time and, thus, yield tractable conjunctive query answering. It is also worthwhile exploring whether other variants of the chase procedure, such as  the semi-oblivious  (a.k.a.\ Skolem) chase~\cite{Marn09,CaPi21} or the core chase~\cite{DeNR08}, can be suitably adapted to our framework so that their desirable properties carry over.

  
Finally, we note that there are several different areas, including  data exchange~\cite{FKMP05}, data integration~\cite{Lenz02,DoHI12}, ontology-mediated query answering~\cite{BiOr15}, and ontology-based data access~\cite{CDLLR18}, 
in which tgds and egds play a crucial role. We believe that the framework presented here makes it possible to infuse entity resolution into these areas in a principled way.

\clearpage

\bibliographystyle{kr}
\bibliography{Bibliography/medium-string,Bibliography/krdb,Bibliography/bibliography,Bibliography/w3c}

\clearpage

\begin{appendices}



\ifthenelse{\equal{\showproofs}{0}}{

\section*{PROOFS OF SECTION 4}

\queryunderhomo*

\partitionofcertainanswers*

%

The following property easily follows from the above proposition.
\begin{corollary}
	\label{cor:partition-of-certain-answers}
	Let $\K$ be a KB, $q$ be a CQ, and 
	$\langle T_1, \dots ,T_n \rangle$ be 
	a certain answer to $q$ with respect to $\K$. If $T_i$ and $T_j$ are sets of entities 
	such that $i \neq j$, then either $T_i = T_j$ or $T_i \cap T_j = \emptyset$.
\end{corollary}

The following property is needed to prove the last theorem of this section. It easily follows from Proposition~\ref{pro:query-under-homo} and Definition~\ref{def:homomorphism}.

\begin{restatable}{corollary}{queryoverunivsolution}
\label{cor:query-over-univ-solution}
Let $\K$ be a KB, $q$ be a CQ, $\U$ be a universal solution for $\K$, and $\vec{\genericSet} \in q^{\U}$. Then, for every $\I \in \Mod(\K)$ there is a tuple $\vec{\genericSet'} \in q^{\I}$ such that  $\vec{\genericSet}\downarrow \leq \vec{\genericSet'}$. 
\end{restatable}

\thmqueryanswering*

\section*{PROOFS OF SECTION 5}

In the sequel, we may write $h(P(\genericSet_1,\ldots,\genericSet_n))$ instead of $P(h(\genericSet_1),\ldots,h(\genericSet_n))$, and, as done in the main body of the paper, may use $h(\tup{\genericSet_1,\ldots,\genericSet_n})$ to denote the tuple $\tup{h(\genericSet_1),\ldots,h(\genericSet_n)}$.

\chaseinterpretation*

\containment*

\chasefinitecasemodel*

\trianglefinitecase*

\chaseunivmodel*

}

\section*{EXAMPLE OF A TERMINATING CHASE SEQUENCE}

We now show a complete execution of the chase procedure in our ongoing example. 

\begin{example}
    \label{ex:chase_full_example}
    Consider the KB $\K=(\T,\DB)$ of Examples~\ref{ex:TBox} and~\ref{ex:DB}, which we report below for the the sake of presentation:
 \begin{align*}
        \small
		\arraycolsep=1.5pt
		\begin{array}{ll}
			(r_1) & \CI(p_1,\name_1,\phone_1) \wedge \CI(p_2,\name_2,\phone_2) \wedge\\
		      &   \JaccardSim(\name_1,\name_2,0.6) \rightarrow p_1 = p_2 \\
			(r_2) & \CI(p, \name_1, \phone_1) \wedge \CI(p, \name_2, \phone_2)\\ 
                 &   \rightarrow \name_1 = \name_2 \\
            (r_3) & \CI(p,\name_1,\phone_1)\wedge \CI(p,\name_2,\phone_2)\\ 
                 &  \rightarrow \phone_1=\phone_2 \\
            (r_4) & \CI(p,\name,\phone) \rightarrow \Emp(p,\comp) \wedge \CEO(\comp,\director)\\
            (r_5) & \Emp(p,\comp_1) \wedge \Emp(p,\comp_2) \rightarrow \comp_1 = \comp_2\\
            (r_6) & \CI(p_1,\name_1,\phone) \wedge \CI(p_2,\name_2,\phone)\\ 
                &  \rightarrow \SameHouse(p_1,p_2,\phone)
		\end{array}
        \normalsize
	\end{align*}
	\begin{align*}
        \small
		\arraycolsep=1.2pt
        \begin{array}{llll}
			(g_1) & \CI(\Doe_1, \texttt{J. Doe}, \PNumberA) &~~~(g_4) & \Emp(\Doe_2, \Yahoo)  \\
			(g_2) & \CI(\Doe_2, \texttt{John Doe}, \PNumberB) &~~~(g_5) & \Emp(\Doe_3, \IBM) \\
			(g_3) & \CI(\Doe_3, \texttt{Mary Doe}, \PNumberA)  &~~~(g_6) & \CEO(\Yahoo, \Doe_1) 
		\end{array}
	\end{align*}

As said, interpretation $\I$ of Example~\ref{ex:interpretation} is the base instance $\I^\DB$ for $\K$. Again, we report it below for convenience.

\small
\[
\arraycolsep=1pt
\begin{array}{ll}
\hspace*{-1pt}(d_1)\CI([\Doe_1], \{ \texttt{J. Doe} \}, \{ \PNumberA \} ) & (d_4)\Emp([\Doe_2], [\Yahoo])\\ 
\hspace*{-1pt}(d_2)\CI([\Doe_2], \{ \texttt{John Doe} \}, \{ \PNumberB \} ) & (d_5)\Emp([\Doe_3], [\IBM])\\
\hspace*{-1pt}(d_3)\CI([\Doe_3], \{ \texttt{Mary Doe} \}, \{ \PNumberA \}) & (d_6)\CEO([\Yahoo],[\Doe_1])
\end{array}
\]
\normalsize
 
    A possible chase sequence is the one composed by the following chase steps:
   
    \smallskip
    \noindent
    \textbf{Step 1:} $\I^\DB \xrightarrow[]{r_4,\mu_0} \I_1$, where $\mu_0$ is such that $\mu_0(\Body(r_4))= \{ d_3 \}$. The resulting instance $\I_1$ contains the facts $d_1$-$d_6$, as well as following facts:
    \begin{align*}
		\arraycolsep=3pt
        \begin{array}{ll}
            (d_7) \Emp([\Doe_3],[\textnormal{\textbf{e}}_1^\bot]), & (d_8) \CEO([\textnormal{\textbf{e}}_1^\bot], [\textnormal{\textbf{e}}_2^\bot]).
		\end{array}
	\end{align*}
    \noindent \textbf{Step 2:} $\I_1 \xrightarrow[]{r_1,\mu_1} \I_2$, where $\mu_1$ is such that $\mu_1(\Body(r_1))= \{ d_1, d_2,\JaccardSim(\{\texttt{J. Doe}\},\{\texttt{John Doe}\},\{0.6\})$. The resulting instance $\I_2$ contains the facts $d_3$, $d_5$, $d_7$, and $d_8$, as well as the following facts:
    \begin{align*}
		\arraycolsep=1.0pt
        \small
		\begin{array}{ll}
			(d_1) \rightarrow (d_9)     & \CI([\Doe_1,\Doe_2], \{ \texttt{J. Doe} \}, \{ \PNumberA \} ) \\
			(d_2) \rightarrow (d_{10}) & \CI([\Doe_1,\Doe_2], \{ \texttt{John Doe} \}, \{ \PNumberB \} ) \\
            (d_4) \rightarrow (d_{11}) & \Emp([\Doe_1,\Doe_2], [\Yahoo]) \\
            (d_6) \rightarrow (d_{12}) & \CEO([\Yahoo],[\Doe_1,\Doe_2])\\
		\end{array}
    \normalsize
	\end{align*}
    
    \smallskip
    \noindent
    \textbf{Step 3:} $\I_2 \xrightarrow[]{r_2,\mu_2} \I_3$, where $\mu_2$ is such that $\mu_2(\Body(r_2))=\{ d_9, d_{10} \}$. The resulting instance $\I_3$ contains the facts $d_3$, $d_{11}$, $d_5$, $d_{12}$, $d_7$, $d_8$, as well as the following facts:
    \begin{align*}
		\arraycolsep=1.0pt
        \small
		\begin{array}{ll}
			(d_9) \rightarrow (d_{13})     & \CI([\Doe_1,\Doe_2], \{ \texttt{J. Doe,John Doe} \}, \{ \PNumberA \} ) \\
			(d_{10}) \rightarrow (d_{14}) & \CI([\Doe_1,\Doe_2], \{ \texttt{J. Doe,John Doe} \}, \{ \PNumberB \} ) \\
    \end{array}
    \normalsize
	\end{align*}
         
    \smallskip
    \noindent
    \textbf{Step 4:} $\I_3 \xrightarrow[]{r_3,\mu_4} \I_4$, where $\mu_3$ is such that $\mu_3(\Body(r_3))= \{ d_{13}, d_{14} \}$. The resulting instance $\I_4$ contains the fact $d_3$, $d_{11}$, $d_5$, $d_{12}$, $d_7$, $d_8$, as well as the following fact:
    \begin{align*}
		\arraycolsep=1.0pt
        \small
		\begin{array}{ll}
			(d_{13}), (d_{14})  \rightarrow (d_{15})   & \CI([\Doe_1,\Doe_2], \{ \texttt{J. Doe,John Doe} \},\\ 
   &~~\{ \PNumberA,\PNumberB \} )
    \end{array}
    \normalsize
	\end{align*}
 Note that $d_{15}$ replaces facts $d_{13}$ and $d_{14}$ of $\I_3$.
    
    \smallskip
    \noindent
    \textbf{Step 5:} $\I_4 \xrightarrow[]{r_5,\mu_5} \I_5$, where $\mu_5$ is such that $\mu_5(\Body(r_5))= \{ d_{5}, d_{7} \}$. The resulting instance $\I_5$ contains facts $d_{14}$, $d_3$, $d_{11}$, $d_{12}$ plus the facts:
    \begin{align*}
		\arraycolsep=1.0pt
        \small
		\begin{array}{ll}
            (d_5), (d_7) \rightarrow (d_{16}) & \Emp([\Doe_3], [\IBM,\textnormal{\textbf{e}}_1^\bot]) \\
            (d_8) \rightarrow (d_{17}) & \CEO([\IBM,\textnormal{\textbf{e}}_1^\bot], [\textnormal{\textbf{e}}_2^\bot])		
    \end{array}
    \normalsize
	\end{align*}
    Note that $d_{16}$ replaces facts $d_{5}$ and $d_{7}$ of $\I_4$.
    
    \smallskip
    \noindent
    \textbf{Step 6:} $\I_5 \xrightarrow[]{r_6,\mu_6} \I_6$, where $\mu_6$ is such that $\mu_6(\Body(r_6))= \{ d_{14}, d_{3} \}$. The resulting instance $\I_6$ contains the same facts as $\I_5$ plus the following fact
    \begin{align*}
		\arraycolsep=3.0pt
		\begin{array}{ll}
			(d_{18}) & \SameHouse([\Doe_1,\Doe_2],[\Doe_3],\{\PNumberA\}) 
		\end{array}
	\end{align*}
 
    \smallskip
    \noindent
    \textbf{Step 7:} $\I_6 \xrightarrow[]{r_6,\mu_7} \I_7$, where $\mu_7$ is such that $\mu_7(\Body(r_6))= \{ d_{14}, d_{3} \}$ (as for step~6, but in this case $ d_{14}$ is the image of the second fact in $\Body(r_6)$ and $d_3$ of the first fact). The resulting instance $\I_7$ contains the same facts as $\I_6$ plus the following fact
    \begin{align*}
		\arraycolsep=3.0pt
		\begin{array}{ll}
			(d_{19}) & \SameHouse([\Doe_3],[\Doe_1,\Doe_2],\{\PNumberA\})  
		\end{array}
	\end{align*}
 
    It is not difficult to see that no rule of $\T$ can be applied to $\I_7$, and thus the chase procedure terminates. The result of the chase is indeed $\I_7$, which coincides with $\I'$ given in Example~\ref{ex:interpretation}. \qed

\end{example}

\end{appendices}

\end{document}